\documentclass[%
 reprint,
 pra,
 amsmath,amssymb,
 aps,
]{revtex4}

\usepackage{graphicx}
\usepackage{dcolumn}
\usepackage{bm}

\usepackage{amsmath,amsfonts,amssymb,dsfont,caption,hyperref,color,epsfig,graphics,graphicx,latexsym,mathrsfs,revsymb,theorem,url,verbatim,epstopdf,cleveref,dashrule,cases,booktabs,diagbox,threeparttable,setspace,natbib,mathdots}

\hypersetup{colorlinks,linkcolor={blue},citecolor={blue},urlcolor={red}}

\newtheorem{definition}{Definition}

\newtheorem{lemma}[definition]{Lemma}

\newtheorem{theorem}[definition]{Theorem}
\newtheorem{corollary}[definition]{Corollary}

\def\squareforqed{$\square$}
\def\qed{\ifmmode\squareforqed\else{\unskip\nobreak\hfil
\penalty50\hskip1em\null\nobreak\hfil\squareforqed
\parfillskip=0pt\finalhyphendemerits=0\endgraf}\fi}

\def\endenv{\ifmmode\;\else{\unskip\nobreak\hfil
\penalty50\hskip1em\null\nobreak\hfil\;
\parfillskip=0pt\finalhyphendemerits=0\endgraf}\fi}

\newenvironment{proof}{\noindent \textbf{{Proof.~} }}{\qed}

\def\bpf{\begin{proof}}
\def\epf{\end{proof}}
\def\bea{\begin{eqnarray}}
\def\eea{\end{eqnarray}}
\def\beq{\begin{equation}}
\def\eeq{\end{equation}}
\def\bal{\begin{aligned}}
\def\eal{\end{aligned}}
\def\bma{\begin{bmatrix}}
\def\ema{\end{bmatrix}}

\def\diag{\mathop{\rm diag}}
\def\tr{{\rm Tr}}

\def\bigox{\bigotimes}
\def\dg{\dagger}

\def\lc{\lceil}
\def\rc{\rceil}
\def\lf{\lfloor}
\def\rf{\rfloor}

\def\op{\oplus}
\def\ox{\otimes}

\def\lin{\mathop{\rm span}}

\def\a{\alpha}
\def\b{\beta}
\def\g{\gamma}

\def\t{\theta}

\def\l{\lambda}
\def\m{\mu}
\def\n{\nu}

\def\r{\rho}
\def\s{\sigma}

\def\ps{\psi}

\def\G{\Gamma}

\newcommand{\bra}[1]{\langle{#1}|}
\newcommand{\ket}[1]{|{#1}\rangle}
\newcommand{\proj}[1]{|{#1}\rangle \langle {#1}|}

\newcommand{\abs}[1]{\left\lvert {#1} \right\rvert}

\newcommand{\etal}{{\sl et~al.}}

\newcommand{\nc}{\newcommand}

\nc{\bbA}{\mathbb{A}} \nc{\bbB}{\mathbb{B}} \nc{\bbC}{\mathbb{C}}
\nc{\bbD}{\mathbb{D}} \nc{\bbE}{\mathbb{E}} \nc{\bbF}{\mathbb{F}}
\nc{\bbG}{\mathbb{G}} \nc{\bbH}{\mathbb{H}} \nc{\bbI}{\mathbb{I}}
\nc{\bbJ}{\mathbb{J}} \nc{\bbK}{\mathbb{K}} \nc{\bbL}{\mathbb{L}}
\nc{\bbM}{\mathbb{M}} \nc{\bbN}{\mathbb{N}} \nc{\bbO}{\mathbb{O}}
\nc{\bbP}{\mathbb{P}} \nc{\bbQ}{\mathbb{Q}} \nc{\bbR}{\mathbb{R}}
\nc{\bbS}{\mathbb{S}} \nc{\bbT}{\mathbb{T}} \nc{\bbU}{\mathbb{U}}
\nc{\bbV}{\mathbb{V}} \nc{\bbW}{\mathbb{W}} \nc{\bbX}{\mathbb{X}}
\nc{\bbY}{\mathbb{Y}} \nc{\bbZ}{\mathbb{Z}}

\nc{\bA}{{\bf A}} \nc{\bB}{{\bf B}} \nc{\bC}{{\bf C}}
\nc{\bD}{{\bf D}} \nc{\bE}{{\bf E}} \nc{\bF}{{\bf F}}
\nc{\bG}{{\bf G}} \nc{\bH}{{\bf H}} \nc{\bI}{{\bf I}}
\nc{\bJ}{{\bf J}} \nc{\bK}{{\bf K}} \nc{\bL}{{\bf L}}
\nc{\bM}{{\bf M}} \nc{\bN}{{\bf N}} \nc{\bO}{{\bf O}}
\nc{\bP}{{\bf P}} \nc{\bQ}{{\bf Q}} \nc{\bR}{{\bf R}}
\nc{\bS}{{\bf S}} \nc{\bT}{{\bf T}} \nc{\bU}{{\bf U}}
\nc{\bV}{{\bf V}} \nc{\bW}{{\bf W}} \nc{\bX}{{\bf X}}
\nc{\bY}{{\bf Y}} \nc{\bZ}{{\bf Z}}

\nc{\bmA}{{\bm A}} \nc{\bmB}{{\bm B}} \nc{\bmC}{{\bm C}}
\nc{\bmD}{{\bm D}} \nc{\bmE}{{\bm E}} \nc{\bmF}{{\bm F}}
\nc{\bmG}{{\bm G}} \nc{\bmH}{{\bm H}} \nc{\bmI}{{\bm I}}
\nc{\bmJ}{{\bm J}} \nc{\bmK}{{\bm K}} \nc{\bmL}{{\bm L}}
\nc{\bmM}{{\bm M}} \nc{\bmN}{{\bm N}} \nc{\bmO}{{\bm O}}
\nc{\bmP}{{\bm P}} \nc{\bmQ}{{\bm Q}} \nc{\bmR}{{\bm R}}
\nc{\bmS}{{\bm S}} \nc{\bmT}{{\bm T}} \nc{\bmU}{{\bm U}}
\nc{\bmV}{{\bm V}} \nc{\bmW}{{\bm W}} \nc{\bmX}{{\bm X}}
\nc{\bmY}{{\bm Y}} \nc{\bmZ}{{\bm Z}}

\nc{\cA}{{\cal A}} \nc{\cB}{{\cal B}} \nc{\cC}{{\cal C}}
\nc{\cD}{{\cal D}} \nc{\cE}{{\cal E}} \nc{\cF}{{\cal F}}
\nc{\cG}{{\cal G}} \nc{\cH}{{\cal H}} \nc{\cI}{{\cal I}}
\nc{\cJ}{{\cal J}} \nc{\cK}{{\cal K}} \nc{\cL}{{\cal L}}
\nc{\cM}{{\cal M}} \nc{\cN}{{\cal N}} \nc{\cO}{{\cal O}}
\nc{\cP}{{\cal P}} \nc{\cQ}{{\cal Q}} \nc{\cR}{{\cal R}}
\nc{\cS}{{\cal S}} \nc{\cT}{{\cal T}} \nc{\cU}{{\cal U}}
\nc{\cV}{{\cal V}} \nc{\cW}{{\cal W}} \nc{\cX}{{\cal X}}
\nc{\cY}{{\cal Y}} \nc{\cZ}{{\cal Z}}

\nc{\hA}{{\hat{A}}} \nc{\hB}{{\hat{B}}} \nc{\hC}{{\hat{C}}}
\nc{\hD}{{\hat{D}}} \nc{\hE}{{\hat{E}}} \nc{\hF}{{\hat{F}}}
\nc{\hG}{{\hat{G}}} \nc{\hH}{{\hat{H}}} \nc{\hI}{{\hat{I}}}
\nc{\hJ}{{\hat{J}}} \nc{\hK}{{\hat{K}}} \nc{\hL}{{\hat{L}}}
\nc{\hM}{{\hat{M}}} \nc{\hN}{{\hat{N}}} \nc{\hO}{{\hat{O}}}
\nc{\hP}{{\hat{P}}} \nc{\hR}{{\hat{R}}} \nc{\hS}{{\hat{S}}}
\nc{\hT}{{\hat{T}}} \nc{\hU}{{\hat{U}}} \nc{\hV}{{\hat{V}}}
\nc{\hW}{{\hat{W}}} \nc{\hX}{{\hat{X}}} \nc{\hZ}{{\hat{Z}}}

\nc{\hn}{{\hat{n}}}

\def\diag{\mathop{\rm diag}}

\def\ine{\mathop{\rm In}}

\def\lin{\mathop{\rm span}}

\def\max{\mathop{\rm max}}

\def\slocc{\mathop{\rm SLOCC}}

\def\tr{\mathop{\rm Tr}}

\def\bigox{\bigotimes}
\def\dg{\dagger}

\def\lc{\lceil}
\def\rc{\rceil}
\def\lf{\lfloor}
\def\rf{\rfloor}
\def\op{\oplus}
\def\ox{\otimes}

\begin{document}

\preprint{APS/123-QED}

\title{Inertias of entanglement witnesses} 

\author{Yi Shen}
\email[]{yishen@buaa.edu.cn}
\affiliation{School of Mathematical Sciences, Beihang University, Beijing 100191, China}
\affiliation{Department of Mathematics and Statistics, Institute for Quantum Science and Technology, University of Calgary, AB, Canada T2N 1N4}


\author{Lin Chen}
\email[]{linchen@buaa.edu.cn (corresponding author)}
\affiliation{School of Mathematical Sciences, Beihang University, Beijing 100191, China}
\affiliation{International Research Institute for Multidisciplinary Science, Beihang University, Beijing 100191, China}

\author{Li-Jun Zhao}
\affiliation{School of Mathematical Sciences, Beihang University, Beijing 100191, China}

\date{\today}

\begin{abstract}
Entanglement witnesses (EWs) are a fundamental tool for the detection of entanglement. 
We study the inertias of EWs, i.e., the triplet of the numbers of negative, zero, and positive eigenvalues respectively. 
We focus on the EWs constructed by the partial transposition of states with non-positive partial transposes.
We provide a method to generate more inertias from a given inertia by the relevance between inertias.
Based on that we exhaust all the inertias for EWs in each qubit-qudit system.
We apply our results to propose a separability criterion in terms of the rank of the partial transpose of state.
We also connect our results to tripartite genuinely entangled states and the classification of states with non-positive partial transposes.
Additionally, the inertias of EWs constructed by X-states are clarified.

\end{abstract}

                             
\maketitle


\section{Introduction}
\label{sec:intro}

Quantum entanglement, discovered by Einstein, Podolsky, Rosen (EPR), and Schr\"odinger \cite{EPR1947,Sch1935}, is a remarkable feature of quantum mechanics. It involves nonclassical correlations between subsystems, and lies in the heart of quantum information theory \cite{entrmp2009,entdect2009}. In recent decades, entanglement has been recognized as a kind of valuable resouce \cite{entrmp2009,rtgour2019,rtent2019}. It plays a central role in various quantum information processing tasks such as quantum computing \cite{qcnature2005}, teleportation \cite{telepor04}, dense coding \cite{densecoding2002}, cryptography \cite{crypt2020}, and quantum key distribution \cite{qkdrmp2020}. 

Although several useful separability criteria such as positive-partial-transpose (PPT) criterion \cite{ppt96,ppte1997}, range criterion \cite{ppte1997}, and realignment criterion \cite{realignment03} were developed, all of them cannot strictly distinguish between the set of entangled states and that of separable ones. According to PPT criterion, any state with non-positive partial transpose (NPT) must be entangled. Nevertheless, the converse only holds for two-qubit and qubit-qutrit systems. There exist PPT entangled (PPTE) states in higher-dimensional Hilbert spaces \cite{ppte1997}.  It has been shown that determining whether a bipartite state is entangled or not is an NP-hard problem \cite{Gurvits03}. It is even harder to tame multipartite entanglement \cite{gme2011}, since multipartite entangled states can be further classified as genuinely multipartite entangled states and biseparable states \cite{geys2020}. 
In 2000, Terhal first introduced the term {\em entanglement witness} (EW) by indicating that a violation of a Bell inequality can be expressed as a witness for entanglement \cite{Terhal1999Bell}. Nowadays, EWs are a fundamental tool for the detection of entanglement both theoretically and experimentally \cite{entdect2009}. 
More and more EWs have been implemented with local measurements \cite{med2019,lewo2020,mdiew2020}.

EWs are observables that enable us to detect entanglement physically. They can be classified as decomposable and non-decomposable EWs. An EW $W$ is decomposable if it can be written as $W=P+Q^\G$ \cite{optew2000}, where $P$ and $Q$ are both positive semidefinite operators, and $Q^\G$ means the partial transpose of $Q$. Otherwise, it is non-decomposable. It is noteworthy that the partial transposition of NPT states is an easy way to construct EWs by the so-called Choi-Jamiolkowski isomorphism \cite{choi_1982}. Furthermore, the partial transpose of an NPT state can be used to construct optimal EWs for decomposable EWs \cite{optew2000}. However, it cannot be directly realized in experiments because the partial transposition is not a physical operation \cite{nielsen00}. To avoid such weakness, the third-moment of the partial transposed density matrix by using bi-local random unitaries was studied in \cite{youzhou2020}. Moreover, in \cite{ptmoments2020} authors proposed and experimentally demonstrated conditions for mixed-state entanglement and measurement protocols based on PPT criterion. These results shed light on the experiments involving the partial transposition. It has been shown that an EW can detect PPTE states if and only if it is non-decomposable \cite{optew2000}. Therefore, much effort has been devoted to constructing non-decomposable EWs \cite{oew2006,pmew2009,oewmap2010,optmap2011}. 

For an NPT state $\rho$, the negative eigenvalues of $\rho^\G$ are a signature of entanglement. They are closely related to other problems in entanglement theory. For instance, the negativity \cite{negativity2002}, a well-known computable entanglement measure, is the sum of the absolute values of negative eigenvalues. Also, by the definition of $1$-distillable state \cite{distill1998}, the more negative eigenvalues $\rho^\G$ has, the more likely $\rho$ is $1$-distillable. Thus, it is important to explore the negative eigenvalues of $\rho^\G$.
The problem of determining how many negative eigenvalues the partial transpose of NPT state could contain has attracted great interest \cite{twoqubitine2008,2xnppt,Neigenv2013,Johnston2013Non}. 
It was first specified in \cite{twoqubitine2008} that $\rho^\G$ has one negative eigenvalue and three positive eigenvalues for any two-qubit entangled state $\rho$. For this reason, an easier method to identify two-qubit entangled states was proposed. That is any two-qubit state is separable if and only if $\det\rho^\G\geq 0$ \cite{twoqubitine2008}. Then some restrictions on the spectral properties of EWs were first derived in \cite{ewspectral2008}.
For NPT state $\rho$ supported on $\mathbb{C}^m\ox\mathbb{C}^n$, it is known that $\rho^\G$ has at most $(m-1)(n-1)$ negative eigenvalues, and all eigenvalues of $\rho^\G$ lie within $[-1/2,1]$ when $\rho$ is normalized \cite{Neigenv2013}. Furthermore, Nathaniel Johnston \etal ~ discussed an interesting problem on the eigenvalues for EWs, namely the inverse eigenvalue problem \cite{inverseew2018}. This problem on EWs inspires us to investigate the matrix inertia of EW instead of considering the number of negative eigenvalues only. We will show the inertia is a finer index to characterize EWs than the number of negative eigenvalues.


In this paper, we study the inertia of EW with a focus on the partial transpose of NPT state in the qubit-qudit system. In the bipartite setting, qubit-qudit states appear in many problems, and have received a lot of attention. Several important properties of qubit-qudit states have been derived. First, all qubit-qudit NPT states are distillable \cite{2ndis2000}. However, the distillability of NPT states in the two-qutrit system still remains as a major open problem in entanglement theory. Second, a first systematic study on the separability of qubit-qudit PPT states was discussed in \cite{2Nsep2000}. Moreover, the birank of qubit-qudit PPT state and the length of qubit-qudit separable state were investigated in \cite{2xnppt}. Very recently, the absolutely separable states in qubit-qudit systems were studied in \cite{2nasep2020} for they are useful in quantum computation. Third, one of the most known analytical formulas for entanglement measures is the entanglement of formation of two-qubit states \cite{eof1998}. Later, a lower bound on entanglement of formation for the qubit-qudit system was derived \cite{2neof2003}. Fourth, the optimization of decomposable EWs acting on the qubit-qudit system was studied \cite{2xnoptdecew2011,optdecewpra2011}. It is known that for a qubit-qudit NPT state $\rho$, $\rho^\G$ is an optimal decomposable EW if and only if the range of $\rho$ contains no product vector \cite{2xnoptdecew2011}.

Here, we introduce a useful tool to study the inertia of EW. That is the inertia of a Hermitian operator is invariant under $\slocc$. Using such a tool we derive the main results in this paper as follows. We first obtain the lower and upper bounds on the number of negative (positive) eigenvalues for an arbitrary bipartite EW in Lemma \ref{le:inertia}. Second we completely determine the inertias of two-qubit EWs in Theorem \ref{cr:twoqubit}. It generalizes the result in \cite{twoqubitine2008}. Third, we show the relation between EWs and the partial transposes of NPT states in Lemma \ref{le:rew>=0}. Then we deeply study the partial transpose of NPT state. In Lemma \ref{le:rho+xid} we reveal the essential relevance between inertias, and propose a method to generate more inertias from a given inertia. This method is also applicable to PPT states. Thus, we can generate inertias for the partial transposes of PPT states as by-products. Moreover, the existence of product vectors in the kernel of $\rho^\G$ is essential to characterize its inertia. Therefore, we discuss this problem in Lemma \ref{le:mxnkerprod}. Based on that we present a sufficient and necessary condition for a sequence to be the inertia in Theorem \ref{le:n-1negative}. Combining Lemma \ref{le:rho+xid} and Theorem \ref{le:n-1negative} we further exhaust all inertias in every qubit-qudit system in Theorem \ref{le:numcN2n-1}. Then we extend our study to general NPT states in Lemma \ref{le:mxnpt}. Finally, we build the connections between our results and other problems in quantum information theory. 
In Theorem \ref{le:newsep} we present a separability criterion in terms of the rank of $\rho^\G$. Then we propose a method to generate the inertia of $\rho^\G$ for higher-dimensional state $\rho$. Using this method we can characterize the inertias of the partial transposes of tripartite genuinely entangled states in a systematic way. We also indicate that the inertia of $\rho^\G$ provides a tool to classify states under SLOCC equivalence. In Theorem \ref{le:xstate} we explicitly express the eigenvalues of $\rho^\G$, and quantify the number of negative ones when $\rho$ is a qubit-qudit X-state \cite{gxstate2010}.

The remainder of this paper is organized as follows. In Sec. \ref{sec:pre} we introduce the preliminaries by clarifying the notations and presenting necessary definitions and useful results. In Sec. \ref{sec:restrict} we show the inertia of an EW is not arbitrary. We first derive some restrictions on the inertia of EW. Second we completely determine the inertia of two-qubit EW. In Sec. \ref{sec:inebiew} we present a sufficient and necessary condition for a sequence to be the inertia, and exhaust all inertias in the qubit-qudit system. In Sec. \ref{sec:app} we show some applications of our results. The concluding remarks are given in Sec. \ref{sec:con}. In the final part, we prove some of our results in the three appendices. In Appendix \ref{sec:proofsec2} we provide the proofs of results in Sec. \ref{sec:restrict}. In Appendix \ref{sec:proof2} we provide the proofs of results in Sec. \ref{sec:inebiew}. In Appendix \ref{sec:proof3} we present the proofs of results in Sec. \ref{sec:app}.


\section{Preliminaries}
\label{sec:pre}

In this section we introduce the preliminaries. First we clarify the notations. Second we introduce some necessary definitions. Finally we present useful results related to the inertia of EW.

We use $\bigox_{i=1}^n \mathbb{C}^{d_i}$ to represent an $n$-partite Hilbert space, where $d_i$'s are local dimensions. If $\rho\in\cB(\bigox_{i=1}^n \mathbb{C}^{d_i})$ is positive semidefinite, then $\rho$ is an $n$-partite state. Unless stated otherwise, the state in this paper is non-normalized. We say $\rho$ is an $m\times n$ state for convenience if $\rho\in\cB(\mathbb{C}^m\ox\mathbb{C}^n)$. Without loss of generality, we may assume $m\leq n$. Since the two partial transposes of $\rho$ with respect to the first and second subsystems respectively are equivalent up to the global transposition, the two partial transposes have the same inertia. Hence, it suffices to consider the partial transpose of $\rho$ with respect to the first system, denoted by $\rho^\G$. For any Hermitian operator $X$, denote by $\cR(X),\cK(X)$, and $r(X)$ the range, kernel and rank of $X$, respectively. Specifically, we will investigate $\cR(\rho^\G),\cK(\rho^\G)$, and $r(\rho^\G)$ for an NPT state $\rho$.
We use $X\geq 0$ to represent a positive semidefinite operator $X$. Denote by $\cM_n$ the set of $n\times n$ matrices, and by $I_n(\in\cM_n)$ the identity matrix. In order to study the inertia of Hermitian $X$ conveniently, we shall refer to the positive (zero, negative) eigen-space of $X$ as the subspace spanned by the eigenvectors corresponding to positive (zero, negative) eigenvalues of $X$. 

In the following we introduce some necessary definitions. In Definition \ref{def:ewdew} we define EWs. The principle of EWs to detect entanglement is depicted in Fig. \ref{fig:EW}. In Definition \ref{def:ine} we define the matrix inertia. In Definition \ref{df:equivalence} we introduce $\slocc$ equivalence. This is a useful tool to study inertias.

\begin{definition}\cite{mew2013}
\label{def:ewdew}
Suppose $W\in\cB(\bigox^n_{i=1} \cH_{i})$ is Hermitian. We call $W$ is an $n$-partite entanglement witness (EW) if (1) it is non-positive semidefinite, and (2) $\bra{\psi}W\ket{\psi}\geq 0$ for any product vector $\ket{\psi}=\bigox^n_{i=1}\ket{a_i}$ with $\ket{a_i}\in\cH_{i}$.
\end{definition}

Suppose $W$ is an $n$-partite EW, and $\rho$ is an $n$-partite state. If $\tr(W\rho)<0$, we determine $\rho$ is an entangled state detected by $W$. On the other hand, the entangled states were indicated as high-level witnesses, i.e., the witnesses for EWs \cite{wew2018}. It is conducive to understand EWs from the perspective of geometry. In convex set theory the {\em Separation Theorem} states that there is a hyperplane separating two disjoint convex sets \cite{cvx2004}. Since the set of all separable states is convex, there is a hyperplane separating the set of all separable states and a subset of entangled ones. Here the hyperplane plays the role of EW. We illustrate how a bipartite EW detects entanglement in Fig. \ref{fig:EW}. 
It is known that a state is entangled if and only if it can be detected by some EW \cite{Terhal1999Bell}. Therefore, the detection of entanglement can be transformed to contructing proper EWs using the positive but not completely positive maps \cite{charew2001}. The transpose map is a typical positive but not completely positive map. This explains why the partial transpose of an NPT state is an EW.


\begin{figure}[ht]
\centering
\includegraphics[width=3in]{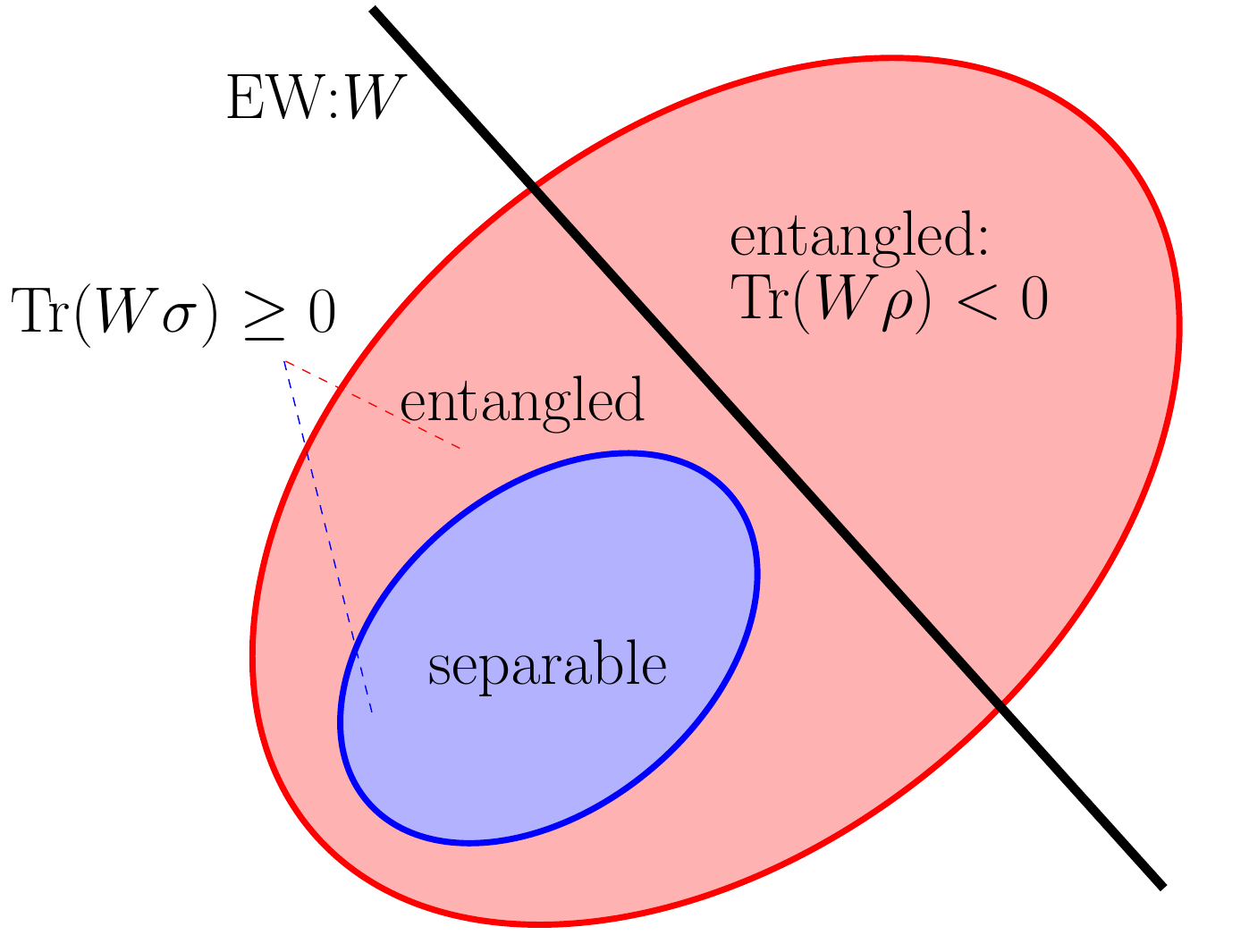}
\caption{The blue oval represents the set of separable states. The red part represents the set of entangled states. The black line represents an EW $W$. It separates the set of all states into upper and lower two parts. For any state $\rho$ in the part above the black line, we conclude that $\tr(W\rho)<0$. Thus this is an entangled state detected by $W$.  For any state $\sigma$ in the part below the black line, we conclude that $\tr(W\rho)\geq 0$. Thus the separability of $\sigma$ cannot be detected by $W$.}
\label{fig:EW}
\end{figure}

In the following we formulate the definition of inertia.

\begin{definition}
\label{def:ine}
Let $A\in\cM_n$ be Hermitian. The inertia of $A$, denoted by $\ine(A)$, is defined as the following sequence
\beq
\label{eq:define}
\ine(A):=(\n_{-},\n_0,\n_{+}),
\eeq
where $\n_{-},\n_0$ and $\n_{+}$ are respectively the numbers of negative, zero and positive eigenvalues of $A$.
\end{definition}

Inertia is an important concept in matrix theory. There is an essential proposition for the matrix inertia, namely Sylvester Theorem \cite{bookma}. It states that Hermitian matrices $A,B\in\cM_n$ have the same inertia if and only if there is a non-singular matrix $S$ such that $B=SAS^\dg$.




Next, we introduce SLOCC equivalence which is an important concept in quantum information theory.

\begin{definition}\cite{SLOCC2000}
\label{df:equivalence}
We refer to SLOCC as stochastic local operations and classical communications.
Two $n$-partite pure states $\ket{\a},\ket{\b}$ are SLOCC equivalent if there exists a product invertible operation $Y=Y_1\ox...\ox Y_n$
such that $\ket{\a}=Y\ket{\b}$.

We further extend the above definitions to spaces. Let $V=\lin\{\ket{\a_1},...,\ket{\a_m}\}$ and $W=\lin\{\ket{\b_1},...,\ket{\b_m}\}$ be two $n$-partite subspaces of $m$-dimension. 
$V$ and $W$ are SLOCC equivalent if there exist a product invertible operation $Y$ such that $\ket{\a_i}\propto Y\ket{\b_i}$ for any $i$.
\end{definition}

Sylvester Theorem implies that inertias are invariant under $\slocc$ equivalence. It allows us to study the inertia under SLOCC equivalence.


In the last part of this section we present several useful results related to the inertia of $\rho^\G$ for NPT state $\rho$. 

\begin{lemma}
\label{le:pure+upper}
Suppose $\r$ is an $m\times n$ NPT state. Then

(i) \cite{Neigenv2013} the number of negative eigenvalues of $\r^\G$ is in the interval $[1,(m-1)(n-1)]$;

(ii) \cite{Neigenv2013} if $\rho$ is normalized, i.e., $\tr(\rho)=1$, every negative eigenvalue of $\r^\G$ is not less than $-\frac{1}{2}$;

(iii) \cite{2xnppt} if $m=2$, for each $k\in[1,n-1]$ there exists a state $\rho$ such that the number of negative eigenvalues of $\rho^\G$ is $k$;

(iv) \cite{inverseew2018} if $\rho$ is a pure state with Schmidt rank $r$, then
\begin{eqnarray}
\label{eq:puresrine}
\ine\r^\G
=
\big(
\frac{r^2-r}{2},
mn-r^2,
\frac{r^2+r}{2}
\big).
\end{eqnarray}
\end{lemma}

Based on the above preliminary knowledge we are ready to study the inertia of EW.

\section{Restrictions on the inertia of entanglement witness}
\label{sec:restrict}

In this section we propose restrictions on the inertia of EW. Specifically, we derive lower and upper bounds on the number of negative (positive) eigenvalues of an EW in Lemma \ref{le:inertia}. In virtue of these restrictions we completely determine the inertia of two-qubit EW in Theorem \ref{cr:twoqubit}. We would like to emphasize that these restrictions will be used frequently in Sec. \ref{sec:inebiew} to further exhaust some inertia sets. Finally we demonstrate the relation between bipartite EWs and bipartite NPT states in Lemma \ref{le:rew>=0}.



First we present the lower and upper bounds on the number of negative (positive) eigenvalues of an EW.

\begin{lemma}
\label{le:inertia}
Suppose $W$ is an EW on $\bbC^m\otimes\bbC^n$.

(i) Let $\cE$ be the non-positive eigen-space of $W$, i.e., the sum of negative and zero eigen-spaces of $W$. Then the product vectors in $\cE$ all belong to the zero eigen-space of $W$. In particular, every vector in the negative eigen-space of $W$ is a pure entangled state.

(ii) The number of negative eigenvalues of $W$ is in $[1,(m-1)(n-1)]$. The decomposable EW containing exactly $(m-1)(n-1)$ negative eigenvalues exists. 

(iii) The number of positive eigenvalues of $W$ is in $[2,mn-1]$. 
\end{lemma}

We give the proof of Lemma \ref{le:inertia} in Appendix \ref{sec:proofsec2}. It is efficient to exclude several sequences to be the inertia of EW by using the restrictions in Lemma \ref{le:inertia}. 


Next, we use Lemma \ref{le:inertia} to determine the inertia of two-qubit EW. In Theorem \ref{cr:twoqubit} we show that every two-qubit EW has inertia $(1,0,3)$. This result generalizes the known conclusion that $\ine(\rho^\G)=(1,0,3)$ for any two-qubit entangled state $\rho$ \cite{twoqubitine2008}.
For this purpose we need to introduce block-positive operators \cite{pmew2009}.
Suppose $M\in\cB(\mathbb{C}^m\ox\mathbb{C}^n)$ is Hermitian. We call $M$ is block-positive, if
\beq
\label{eq:defblockpos}
M:= (I_m\ox\Phi) X,
\eeq
for some positive semidefinite operator $X\in\cB(\mathbb{C}^m\ox\mathbb{C}^m)$, and some positive map $\Phi:\cB(\mathbb{C}^m)\to\cB(\mathbb{C}^n)$.
It is known that $W$ is an EW if and only if it is block-positive but non-positive semidefinite \cite{pmew2009}. 
In \cite{inverseew2018} there was a useful result on the eigenvalues of block-positive operators in $\cB(\bbC^2\otimes \bbC^2)$. It states that
there exists a block-positive matrix $W$ on $\bbC^2\otimes \bbC^2$ with eigenvalues $\mu_1\geq \mu_2\geq \mu_3\geq \mu_4$  if and only if the following three inequalities hold:
\begin{eqnarray}\label{eig:iff}
\begin{aligned}
\mu_3&\geq 0,\\\mu_4&\geq -\mu_2,\\\mu_4&\geq-\sqrt{\mu_1 \mu_3}.
\end{aligned}
\end{eqnarray}
Combining Lemma \ref{le:inertia} and the above result we can show Theorem \ref{cr:twoqubit} as follows.
\begin{theorem}
\label{cr:twoqubit}
Every two-qubit EW has inertia $(1,0,3)$.	
\end{theorem}
\begin{proof}
	By Lemma \ref{le:inertia} (ii), any two-qubit EW has exact one negative eigenvalue. Thus, there are two distributions of inertia $(1,0,3)$ and $(1,1,2)$. Here we prove that sequence $(1,1,2)$ is not the inertia. Assume $W$ is a two-qubit EW with inertia $(1,1,2)$. Let $\mu_1\geq \mu_2\geq \mu_3\geq \mu_4$ be the four eigenvalues of $W$. It follows that $\mu_4$ is negative and $\mu_3=0$. It contradicts with the last inequality in \eqref {eig:iff}. So the assumption is not valid. Therefore, the inertia $(1,1,2)$ does not exist. This completes the proof.
\end{proof}

Theorem \ref{cr:twoqubit} motivates us to determine the inertias of EWs acting on higher-dimensional Hilbert spaces.
As we know, the partial transpose of NPT state is an EW. Obviously, there are EWs which are not the partial transpose of NPT state. Here we construct an example to show that there exists an EW $W$ whose partial transpose $W^\G$ is still an EW. Thus, $W$ cannot be the partial transpose of NPT state. Let $\a=(\ket{00}+\ket{11})(\bra{00}+\bra{11})$ and $\b=\proj{00}+a\proj{11}+b(\ket{01}+\ket{10})(\bra{01}+\bra{10})+c(\ket{01}-\ket{10})(\bra{01}-\bra{10})$ with
\beq
\bal
&a,b>0, \quad c\in(0,1/2),\\
&2(1 + a) - (1 + b - c)^2 < 0.	
\eal
\eeq
One can verify that $W=\a^\G+\b$ is an EW, and $W^\G$ is still an EW. Inspired by this example, we demonstrate the relation between bipartite EWs and bipartite NPT states in Lemma \ref{le:rew>=0}.

\begin{lemma}
\label{le:rew>=0}
Suppose $W\in\cB(\mathbb{C}^m\otimes\mathbb{C}^n)$ is a Hermitian and non-positive semidefinite operator. Then $W$ is an EW if and only if $W^\G$ is an EW or an NPT state. 
\end{lemma}

\begin{proof}
Let $\cT$ be the set of bipartite EWs and NPT states. We first show $\cT$ is invariant under partial transpose. Suppose $W\in\cT$. If $W$ is an EW then $W^\G$ is still Hermitian. Further if $W^\G$ is positive semidefinite then $W^\G$ is indeed an NPT state. Thus we conclude that $W^\G\in\cT$. If $W^\G$ is non-positive semidefinite one can show $W^\G$ is still an EW as follows. For any product vector $\ket{a_1,a_2}$,
$$
\bra{a_1,a_2}W^\G\ket{a_1,a_2}=\bra{a_1^*,a_2}W\ket{a_1^*,a_2}\geq 0.
$$
Thus we conclude that $W^\G\in\cT$. For the same reason we conclude that $W^\G\in\cT$ if $W$ is an NPT state. Therefore, $\cT$ is invariant under partial transpose. This result implies that $W$ is an EW if and only if $W^\G$ is an EW or an NPT state. This completes the proof.
\end{proof}

In experiments, an EW is usually decomposed into a sum of locally measurable observables. Then these locally measurable observables are measured individually on the constituent subsystems. Finally one obtains {\em witness expectation value} $\tr(W\rho)$ by summing the expectation values of the locally measurable observables. In \cite{dectent2002} O. G\"uhne \etal ~ introduced a general method for the experimental detection of entanglement by performing only few local measurements, assuming some prior knowledge of the density matrix. Their method is based on the minimal decomposition of witness operators into a pseudomixture of local operators. Any bipartite EW $W$ can be decomposed into a sum of projectors onto product vectors, i.e.,
\beq
\label{eq:EWdec}
W=\sum_j c_j \proj{a_j,b_j}=\sum_j c_j \proj{a_j}\ox\proj{b_j},
\eeq
where the coefficients $c_j$ are real and satisfy $\sum_j c_j=1$. There is at least one coefficient has to be negative for $W$ is an EW. This characterizes a so-called {\em{pseudomixture}}. 
Moreover, there are many different decompositions like in Eq. \eqref{eq:EWdec} for any EW. In \cite{dectent2002} authors were interested in the optimal decompositions. That is the pseudomixture with minimal number of non-zero coefficients $c_j$. Suppose $W$ is an EW which is not the partial transpose of NPT state. It follows from Lemma \ref{le:rew>=0} that $W^\G$ is also an EW. One can verify that if Eq. \eqref{eq:EWdec} is a decomposition of $W$, then 
\beq
\label{eq:EWdec-2}
W^\G=\sum_{j} c_j \proj{a_j^*,b_j}
\eeq
is a decomposition of $W^\G$. It implies that the minimal number of non-zero coefficients for $W$ is the same as that for $W^\G$.  

In the following section we investigate the inertia of EW starting from the EWs constructed by the partial transpose of NPT state. 
Due to the relation given by Lemma \ref{le:rew>=0} our results are helpful to understand the inertia of general EW.

\section{Inertias of the partial transposes of NPT states}
\label{sec:inebiew}

The partial transposition on an NPT state is an easy way to construct EWs and can be used to construct optimal EWs for decomposable EWs \cite{optew2000}. In this section we focus on the bipartite EWs constucted by the partial transpose of NPT state, and determine inertias of such EWs. 
In Lemma \ref{le:rho+xid} we reveal the essential relevance between inertias, and propose a method to generate more inertias from a given inertia. We apply Lemma \ref{le:rho+xid} to NPT states in qubit-qudit systems. The qubit-qudit states are widely investigated and have many interesting propositions. Suppose $\rho$ is a $2\times n$ NPT state. In Theorem \ref{le:n-1negative} we show a sufficient and necessary condition for a sequence $(a,b,c)$ to be an inertia of $\rho^\G$. Based on the above results we exhaust all inertias for $\rho^\G$ in Theorem \ref{le:numcN2n-1}. Finally, in Lemma \ref{le:mxnpt} we extend our study to $m\times n$ NPT states.


In the first part of this section we focus on the partial transposes of all states, though in this paper we are more interested in NPT states. In order to describe our results conveniently, we first denote three inertia sets: 
\beq
\label{eq:definesetmn}
\bal
\cN_{m,n}&:=\{\ine(\rho^\G)|\text{$\rho$ is an $m\times n$ NPT state.}\},\\
\cP_{m,n}&:=\{\ine(\rho^\G)|\text{$\rho$ is an $m\times n$ PPTE state.}\},\\
\cS_{m,n}&:=\{\ine(\rho^\G)|\text{$\rho$ is an $m\times n$ separable state.}\}.\\
\eal
\eeq

In the following we propose an effective method to derive more inertias from a given inertia. It also reveals the relevance between inertias regarding the existence of product vectors in the kernel of $\rho^\G$.

\begin{lemma}
\label{le:rho+xid}
(i) Suppose $\r$ is an $m\times n$ NPT (PPTE, separable) state and $\r^\G$ has inertia $(a,b,c)$. Then there is a small enough $x>0$ and the NPT (PPTE, separable) state 
$$\s:=\r+xI_{mn},$$
such that 
$$\ine(\s^\G)=(a,0,b+c).$$
Note that if $\rho$ is PPT, then $a=0$.

(ii) Suppose $m_1\leq m_2$ and $n_1\leq n_2$. If 
\beq
\label{eq:inertiatree-1}
\bal
(a_1,b_1,c_1)&\in\cN_{m_1,n_1},\\
\text{or} ~ (a_1,b_1,c_1)&\in\cP_{m_1,n_1},\\
\text{or} ~ (a_1,b_1,c_1)&\in\cS_{m_1,n_1},
\eal
\eeq
with $a_1+b_1+c_1=m_1n_1$, then $\forall~ 0\leq l\leq m_2n_2-m_1n_1$,
\beq
\label{eq:inertiatree-2}
\bal
(a_1,m_2n_2-m_1n_1-l,b_1+c_1+l)&\in\cN_{m_2,n_2},\\
\text{or} ~ (a_1,m_2n_2-m_1n_1-l,b_1+c_1+l)&\in\cP_{m_2,n_2},\\
\text{or} ~ (a_1,m_2n_2-m_1n_1-l,b_1+c_1+l)&\in\cS_{m_2,n_2},
\eal
\eeq
respectively.
Note that if $(a_1,b_1,c_1)\in\cP_{m_1,n_1}$ or $(a_1,b_1,c_1)\in\cS_{m_1,n_1}$, then $a_1=0$.
\end{lemma}

We show the proof of Lemma \ref{le:rho+xid} in Appendix \ref{sec:proof2}.
Using this idea one can imagine how inertias grow as local dimensions increase. We illustrate this growing process in Fig. \ref{fig:fam}.
The basic idea of Lemma \ref{le:rho+xid} is to add linearly independent product states into the given density matrix. We will apply this method to further characterize the inertia set $\cN_{2,n}$. 

\begin{figure}[ht]
\centering
\includegraphics[width=3in]{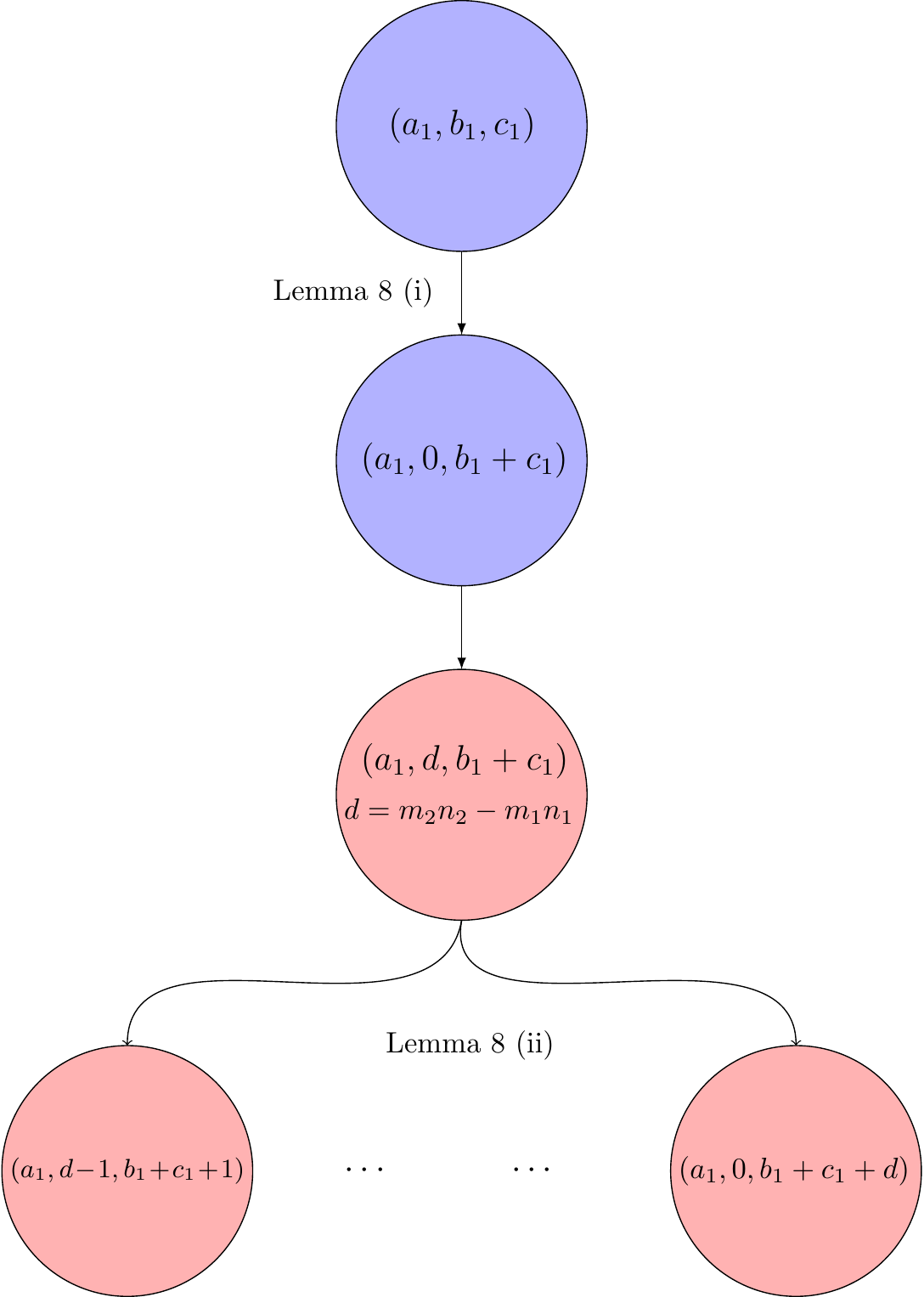}
\caption{The blue circle represents $\cN_{m_1,n_1}$, and the red circle represents $\cN_{m_2,n_2}$, where $m_1\leq m_2$ and $n_1\leq n_2$. Let $d=m_2n_2-m_1n_1$. The sequence in blue or red circle means it belongs to the corresponding inertia set. To conclude $(a_1,d,b_1+c_1)\in\cN_{m_2,n_2}$, it suffices to add proper zero rows and zero columns into the density matrix of $\rho$ which satisfies $\ine(\rho^\G)=(a_1,0,b_1+c_1)$.}
\label{fig:fam}
\end{figure}

In the second part of this section we aim to determine the inertia set $\cN_{2,n}$ completely. There are two main results in this part. One is Theorem \ref{le:n-1negative}, where we propose a sufficient and necessary condition for a sequence in the inertia set $\cN_{2,n}$. The other one is Theorem \ref{le:numcN2n-1}, where we completely determine the inertia set $\cN_{2,n}$ for any $n\geq 2$.

It follows from Lemma \ref{le:inertia} (i) that all product vectors in the non-positive eigen-space of an EW belong to the kernel of this EW. The existence of product vectors in $\cK(\rho^\G)$ is quite essential for studying $\ine(\rho^\G)$. For this reason we investigate how many linearly independent product vectors in $\cK(\r^\G)$ as follows. 
\begin{lemma}
\label{le:mxnkerprod}
Let $\r$ be an $m\times n$ NPT state. Denote by $d$ the dimension of $\cK(\r^\G)$. Suppose $\rho^\G$ has $k$ negative eigenvalues, and $d+k>(m-1)(n-1)$. Let $l=d+k-(m-1)(n-1)$. Then there are at least $l$ linearly independent product vectors in $\cK(\r^\G)$. That is
\beq
\label{eq:kerprodspan-1}
\cK(\r^\G)=\lin\{\ket{a_1,b_1},\cdots,\ket{a_{l},b_{l}},\ket{u_{l+1}},\cdots,\ket{u_d}\}.
\eeq
Specifically, if $k=(m-1)(n-1)$, then $l=d$ which implies $\cK(\r^\G)$ is spanned by product vectors.
\end{lemma}

We present the proof of Lemma \ref{le:mxnkerprod} in Appendix \ref{sec:proof2}.
For an $m\times n$ NPT state $\rho$, it follows from Lemma \ref{le:inertia} (ii) that the number of negative eigenvalues of $\rho^\G$ is not greater than $(m-1)(n-1)$, i.e., $k\leq (m-1)(n-1)$ in Lemma \ref{le:mxnkerprod}. Therefore, if $d+k>(m-1)(n-1)$, then $d>0$, and thus $\cK(\rho^\G)$ is not a zero space.

From Lemma \ref{le:mxnkerprod} we can determine whether there exist product vectors in $\cK(\rho^\G)$ based on the dimension of the non-positive eigen-space of $\rho^\G$. The existence of product vectors in the kernel is useful to simplify the problem of determining inertias.
In the following we apply the above results to $2\times n$ NPT states. In Theorem \ref{le:n-1negative} we propose a sufficient and necessary condition for a sequence $(a,b,c)$ in the inertia set $\cN_{2,n}$.
\begin{theorem}
\label{le:n-1negative}
Suppose $a$ is a positive integer and $b,c$ are non-negative integers with $a+b+c=2n$ and $a+b>n-1$. 

(i) Then $(a,b,c)\not\in\cN_{2,n}$ if and only if the following two conditions are satisfied.

(i.a) $(a,b-2,c)\not\in\cN_{2,n-1}$,

(i.b) $(a,b-1,c-1)\not\in\cN_{2,n-1}$. 

(ii) Suppose $\ine(\rho^\G)=(a,b,c)$. If $(a,b-2,c)\in\cN_{2,n-1}$, and $(a,b-1,c-1)\not\in\cN_{2,n-1}$, then $\rho$ can be regarded as a $2\times (n-1)$ NPT state up to a local projector.
\end{theorem}

We show the proof of Theorem \ref{le:n-1negative} in Appendix \ref{sec:proof2}. Theorem \ref{le:n-1negative} demonstrates the relation between the two inertia sets $\cN_{2,n-1}$ and $\cN_{2,n}$ for any $n>2$.
Applying this result we obtain the following corollary.

\begin{corollary}
\label{th:2x3ine}

(i) There exists a $2\times n$ NPT state $\rho$ whose partial transpose contains exact $(n-1)$ negative eigenvalues. Further, if $\rho^\G$ has $(n-1)$ negative eigenvalues, then $\ine (\rho^\G)=(n-1,0,n+1)$.

(ii) We determine the inertia set $\cN_{2,3}$ as follows.
\beq
\label{eq:2x3ptineexist}
\cN_{2,3}=\{(1,2,3),(1,1,4),(1,0,5),(2,0,4)\}.
\eeq

(iii) Suppose $\rho_{AB}$ is a $2\times n$ NPT state. For any $j\in[1,n-1]$, if $\ine(\rho_{AB}^\G)=\big(j,2(n-1-j),j+2\big)$, then $r(\rho_B)=j+1$, i.e., $\rho_{AB}$ is indeed a $2\times(j+1)$ NPT state up to a local projector.
\end{corollary}

We present the proof of this corollary in Appendix \ref{sec:proof2}.
Combining Lemma \ref{le:rho+xid} and Theorem \ref{le:n-1negative}, and using mathematical induction we can further exhaust $\cN_{2,n}$ for any $n\geq 2$. We will discuss the details in Theorem \ref{le:numcN2n-1}.



\begin{theorem}
\label{le:numcN2n-1}
There are exact $(n-1)^2$ distinct inertias in $\cN_{2,n}$, i.e.,
\beq
\label{eq:numcN2n-1}
\abs{\cN_{2,n}}=(n-1)^2,~\forall n\geq 2.
\eeq 
Furthermore, the $(n-1)^2$ distinct inertias in $\cN_{2,n}$ are as follows.
\beq
\label{eq:numcN2n-3}
\bal
&(1,2(n-2)-j,j+3),\quad \text{$\forall 0\leq j\leq 2(n-2)$},\\
&(2,2(n-3)-j,j+4),\quad \text{$\forall 0\leq j\leq 2(n-3)$},\\
&\vdots\\
&(n-1,0,n+1).
\eal
\eeq
\end{theorem}

We provide the proof of Theorem \ref{le:numcN2n-1} in Appendix \ref{sec:proof2}. By Theorem \ref{le:numcN2n-1} we completely determine the inertia set $\cN_{2,n}$ for any $n\geq 2$. Using the method in Lemma \ref{le:rho+xid} one can construct the example whose partial transpose has the corresponding inertia in \eqref{eq:numcN2n-3}.
An observation from \eqref{eq:numcN2n-3} is that if $\rho$ is a $2\times n$ NPT state, then $\rho^\G$ has at least one negative and three positive eigenvalues. Based on this observation we prove that any bipartite NPT state shares this property using mathematical induction.

\begin{lemma}
\label{le:mxnpt}
If $\rho$ is an $m\times n$ NPT state for any $m,n\geq 2$, then $\rho^\G$ has at least one negative and three positive eigenvalues. Furthermore, $(1,mn-4,3)\in\cN_{m,n}$ for any $m,n\geq 2$.
\end{lemma}

\begin{proof}
It follows from Lemma \ref{le:inertia} (ii) that $\rho^\G$ has at least one negative eigenvalue. Hence, we only need to show $\rho^\G$ has at least three positive eigenvalues. We use mathematical induction to prove it. First, $\rho^\G$ has the property for $m=2$ and any $n\geq 2$ from \eqref{eq:numcN2n-3}. Second we assume $\rho^\G$ has the property for $m=k(\geq 2)$ and any $n\geq 2$. Finally we prove $\rho^\G$ has the property for $m=k+1$ and any $n\geq 2$. We prove it by contradiction. From Lemma \ref{le:inertia} (iii) it suffices to denote $\ine\rho^\G=(a,b,2)$, where $a+b=(k+1)n-2$ and $1\leq a\leq k(n-1)$. It follows that 
$$
n-2+k\leq b\leq (k+1)n-3.
$$
Since $k\geq 2$ by assumption, we conclude that $b\geq n$. It implies that $(a,b-n,2)\in\cN_{k,n}$. However, it contradicts with the induction hypoethesis that $\rho^\G$ has at least three positive eigenvalues for $m=k$ and any $n\geq 2$. Therefore, we conclude that $\rho^\G$ has at least three positive eigenvalues for $m=k+1$ and any $n\geq 2$. Thus, by mathematical induction our claim holds. For the last assertion, since $\cN_{2,2}=\{(1,0,3)\}$, it follows that $(1,mn-4,3)\in\cN_{m,n}$ for any $m,n\geq 2$.
This completes the proof.
\end{proof}

Lemma \ref{le:mxnpt} partially improves Lemma \ref{le:inertia} (iii). We restrict EWs here into the partial transpose of NPT state. It is interesting to ask whether all bipartite EWs share this property that the number of positive eigenvalues is at least three. It is related to the EWs with the minimal rank.
A direct corollary from Lemma \ref{le:mxnpt} is that if $\rho$ is an $m\times n$ NPT state for any $m,n\geq 2$, then $\rho^\G$ has rank at least four. It can be used to construct a separability criterion.

\section{Connections with other problems}
\label{sec:app}

In this section we build the connections between the inertia of EW and other aspects in quantum information theory. 
First we present a separability criterion based on the rank of $\rho^\G$ in Theorem \ref{le:newsep}. 
Second we propose a method to generate the inertia of $\rho^\G$ for higher-dimensional state $\rho$. Using this method we can characterize the partial transpose of tripartite genuinely entangled state in a systematic way. 
Third we indicate that the inertia of $\rho^\G$ provides a tool to classify states under SLOCC equivalence.
Fourth when $\rho$ is a $2\times n$ X-state, we explicitly express the eigenvalues of $\rho^\G$, and quantify the number of negative ones in Theorem \ref{le:xstate}.

First, determining whether a state is entangled or separable is a central and long-standing problem in entanglement theory \cite{entrmp2009}. PPT criterion is commonly used, while it is necessary but not sufficient for high dimensional states \cite{ppte1997}. Thus, the separability of PPT state with small rank has been investigated \cite{mxnppt2000,mxndistill2003,seplin2013}. 
In the following theorem we present a separability criterion for the states whose partial transposes have small ranks.
\begin{theorem}
\label{le:newsep}
Suppose $\rho$ is an $n$-partite state. Denote by $\rho^{\G_S}$ the partial transpose of $\rho$ with resepct to the subsystem $S\subseteq\{1,\cdots,n\}$. If for any subsystem $S$, $\rho^{\Gamma_S}$ has rank at most three, then $\rho$ and $\rho^{\G_S}$ are both separable.
\end{theorem}

\begin{proof}
Let $S^c$ be the complement of $S$ in $\{1,\cdots,n\}$. First we take $\rho$ as a bipartite state of system $S,S^c$.
It follows from Lemma \ref{le:mxnpt} that if $\rho$ is a bipartite NPT state, then $r(\rho^{\G_S})\geq 4$. Thus if $r(\rho^{\G_S})\leq 3$, then $\rho$ is a bipartite PPT state in the bipartition $S|S^c$. Thus, $\rho^{\G_S}$ is positive semidefinite, and indeed a bipartite PPT state in the bipartition $S|S^c$. Therefore, if for any subsystem $S$, $\rho^{\Gamma_S}$ has rank at most three, it implies $\rho^{\G_S}$ is PPT in any bipartition. Thus, for any subsystem $S$, $\rho^{\G_S}$ is an $n$-partite PPT state. It is known that any multipartite PPT state of rank at most three is separable \cite{seplin2013}. Hence, for any subsystem $S$, $\rho^{\G_S}$ is separable, and thus $\rho$ is also separable. This completes the proof.
\end{proof}

The advantage of Theorem \ref{le:newsep} is that we don't need to check whether $\rho$ is a PPT state if with respect to each bipartition the partial transpose has rank at most three. The reason is that Lemma \ref{le:mxnpt} guarantees such states are multipartite PPT ones. However, if $r(\rho^\G)$ is greater than three, one cannot determine whether $\rho$ is a PPT state directly. Therefore, in a similar way it is possible to propose other useful separability criteria from the perspective of bi-rank $(r(\rho),r(\rho^\G))$ when we fully characterize the inertia set of $\rho^\G$ for general NPT state $\rho$.


Second, we propose a method to generate the inertia of $\rho^\G$ in higher-dimensional systems. The method is depicted in Fig. \ref{fig:high}. 
Suppose $\a_{AB}$ is an $m_1\times n_1$ NPT state of system $A,B$, and $\b_{CD}$ is an $m_2\times n_2$ NPT state of system $C,D$. Denote $\ine(\a_{AB}^\G)=(a_1,b_1,c_1)$ and $\ine(\b_{CD}^\G)=(a_2,b_2,c_2)$. Let $\rho_{(AC):(BD)}:=\a_{AB}\ox \b_{CD}$ be a bipartite state of system $(AC),(BD)$. Then $\rho_{(AC):(BD)}^{\G}$ has the inertia
\beq
\label{eq:kroneckerine}
(a_1c_2+a_2c_1, b_1m_2n_2+b_2m_1n_1-b_1b_2, a_1a_2+c_1c_2).
\eeq
The inertia \eqref{eq:kroneckerine} can be verified directly. Since 
$$\rho_{(AC):(BD)}^\G=\a_{AB}^{\G}\ox \b_{CD}^\G,$$ 
the number of negative eigenvalues is $a_1c_2+a_2c_1$, and the number of positive eigenvalues is $a_1a_2+c_1c_2$.

\begin{figure}[ht]
\centering
\includegraphics[width=3.5in]{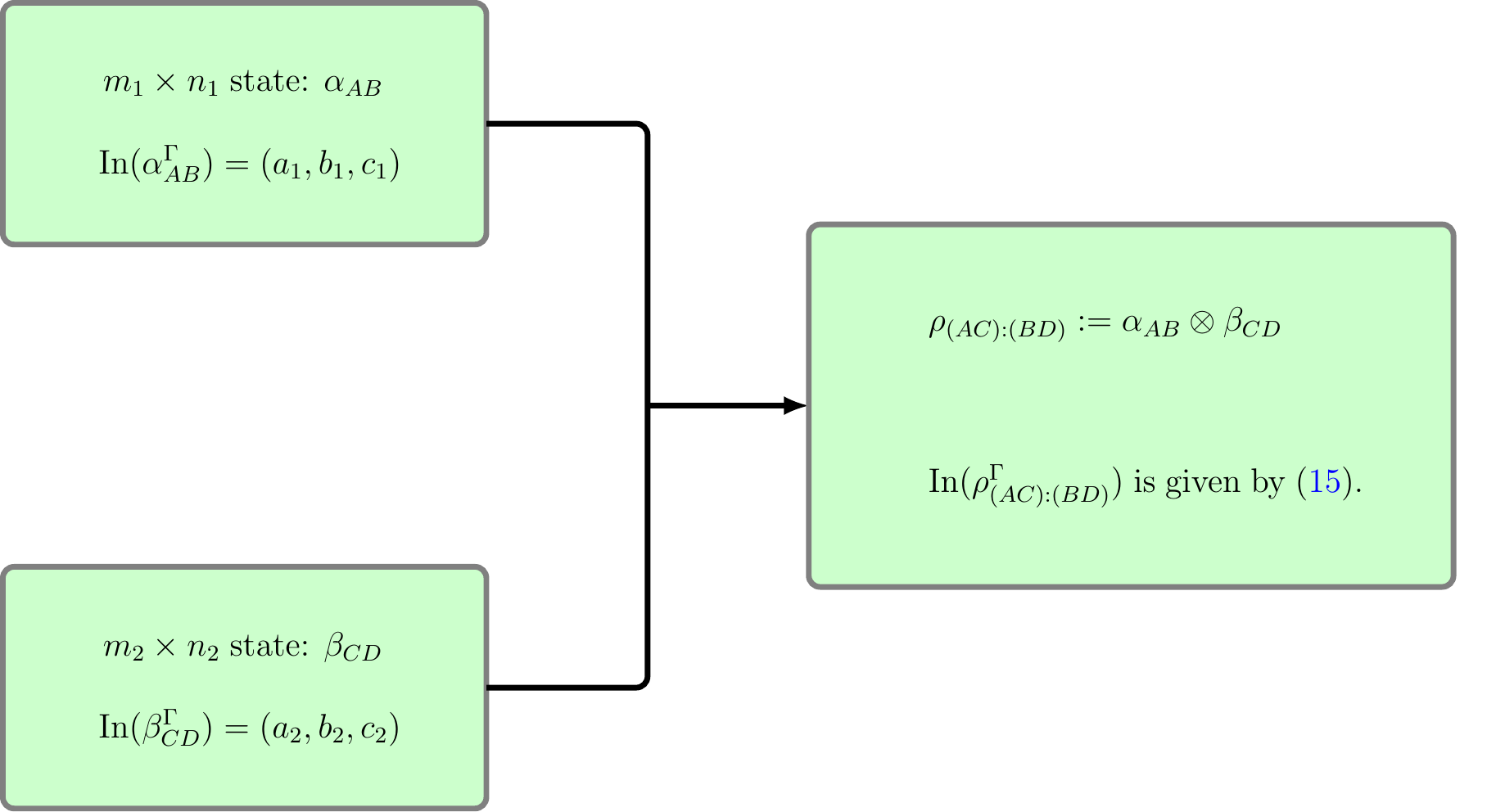}
\caption{Here, $\a_{AB}$ is an $m_1\times n_1$ state with $\ine(\a_{AB}^\G)=(a_1,b_1,c_1)$, and $\b_{CD}$ is an $m_2\times n_2$ state with $\ine(\b_{CD}^\G)=(a_2,b_2,c_2)$. Then we construct an $m_1m_2\times n_1n_2$ state of system $(AC),(BD)$. The inertia of $\rho_{(AC):(BD)}^\G$ is given by \eqref{eq:kroneckerine}. } 
\label{fig:high}
\end{figure}

By splitting system $(BD)$ into two subsystems $B,D$, we can take $\rho_{(AC):(BD)}$ as a tripartite state of system $(AC),B,D$, i.e., $\rho_{(AC):B:D}$, and take $\rho_{(AC):(BD)}^\G$ as the partial transpose of the tripartite state $\rho_{(AC):B:D}$ with respect to subsystem $(AC)$, i.e., $\rho_{(AC):B:D}^{\G_{AC}}$. In this way we can construct the tripartite genuinely entangled state $\rho_{(AC):B:D}$ using two bipartite entangled states $\a_{AB}$ and $\b_{CD}$ \cite{geys2020}. Moreover, we conjectured in \cite{geys2020} that $\rho_{(AC):B:D}$ is a tripartite genuinely entangled state if both $\a_{AB}$ and $\b_{CD}$ are entangled. We have shown the above conjecture holds if either $\cR(\a_{AB})$ or $\cR(\b_{CD})$ is not spanned by product vectors \cite{geys2020}. The latest progress on this conjecture has been made in \cite{sun2020}. As we know, genuine mulitpartite entanglement is valuable resouce in quantum information processing tasks \cite{12qubitge2019,rtent2019,gme2020}. Nevertheless, it is difficult to characterize genuinely multipartite entangled (GME) states \cite{gme2011,geys2020}. Obviously, the characterization of the partial transpose of GME state is also hard. As far as we know, there are few papers discussing the inertia of the partial transpose of GME state. If the above-mentioned conjecture is true, using the method in Fig. \ref{fig:high} we find a systematic way to construct tripartite genuinely entangled states whose partial transposes have inertias that may be exhausted explicitly. For example, if $\a_{AB}$ and $\b_{CD}$ are two $2\times n$ NPT states, we can exhaust the inertia of $\rho_{(AC):B:D}^\G$ by Theorem \ref{le:numcN2n-1}. Furthermore, we can construct a tripartite genuinely entangled state whose partial transpose has a given inertia in this way.

Third, we indicate that the inertia of $\rho^\G$ can be used to classify states under SLOCC equivalence. In quantum information theory, the classification of multipartite states is one of the central problems and has received extensive attentions in the past decades \cite{3qubitlu2000,3qubitinequiv2001,luequiv2010,mpsinequiv2013}. Two main approaches of classication are the equivalence under local unitary (LU) and SLOCC operations \cite{SLOCC2000}. For example, a complete classification of pure three-qubit states in terms of LU equivalence were presented in \cite{3qubitlu2000}. In terms of SLOCC equivalence, it has been shown that only two inequivalent classes for pure three-qubit genuinely entangled states, namely the {\em W}-state class and {\em GHZ}-state class \cite{3qubitinequiv2001}. Moreover, necessary and sufficient conditions for the equivalence of arbitrary $n$-qubit pure quantum states under LU operations were derived in \cite{luequiv2010}. A systematic classification of multiparticle entanglement in terms of SLOCC equivalence were provided in \cite{mpsinequiv2013}.

In the following we introduce a classification of $m\times n$ NPT states using the inertias of their partial transposes. The inertia is invariant under SLOCC operations from Sylvester Theorem. Moreover, we conclude that if two $n$-partite mixed states of system $A_1,...,A_n$ are SLOCC equivalent, then their partial transposes with respect to any $k$-partite subsystem $A_{j_1},....,A_{j_k}$ are SLOCC equivalent. (We prove this claim, i.e., Lemma \ref{le:sloccequiv} (i), in Appendix \ref{sec:proof3}.) As a result, if $\rho_{AB}^\G$ and $\sigma_{AB}^\G$ have different inertias, then $\rho_{AB}$ and $\sigma_{AB}$ are SLOCC inequivalent. Therefore, we propose a necessary condition for $\rho_{AB}$ and $\sigma_{AB}$ to be SLOCC equivalent, i.e., $\ine(\rho_{AB}^\G)=\ine(\sigma_{AB}^\G)$. Further, for $2\times n$ NPT states, from Theorem \ref{le:numcN2n-1} we conclude that there are at least $(n-1)^2$ inequivalent families in terms of SLOCC equivalence.

Furthermore, we introduce the concept of strong SLOCC inequivalence. Suppose $\rho_{AB}$ and $\sigma_{AB}$ are both $2\times n$ NPT states. We consider the $N$ copies of $\rho_{AB}$ and $\sigma_{AB}$, i.e.,  $\rho_{AB}^{\ox N}$ and $\sigma_{AB}^{\ox N}$. We find that if $\rho_{AB}^\G$ and $\sigma_{AB}^\G$ have different inertias, then the partial transposes of $\rho_{AB}^{\ox N}$ and $\sigma_{AB}^{\ox N}$ still have different inertias. (We prove this claim, i.e., Lemma \ref{le:sloccequiv} (ii), in Appendix \ref{sec:proof3}.) We call this relation strong SLOCC inequivalence. Physically, it implies that the collective use of many copies cannot change the inequivalence under SLOCC. The classification of states enables us to determine whether there exist SLOCC operations to transform a state to another one. The transformation between many copies of two pure multipartite states was studied in \cite{MCSLOCC2011}. It has been shown that two transformable multipartite states under SLOCC are also transformable under multicopy SLOCC \cite{MCSLOCC2011}. The  strong SLOCC inequivalence here shows that if $\rho_{AB}$ and $\sigma_{AB}$ cannot be transformed under SLOCC, then $\rho_{AB}^{\ox N}$ and $\sigma_{AB}^{\ox N}$ cannot be transformed under multicopy SLOCC too.

Fourth, we discuss a class of states called X-states. They are defined as states whose density matrix has nonzero elements only along its diagonal and antidiagonal in resemblance to the letter $\bf X$ \cite{gxstate2010}. For example, GHZ diagonal states are typical kinds of X-states \cite{ghzdiag2017}. X-states are important states that occur in various contexts such as entanglement \cite{xchen2017}, its decay under decoherence \cite{xdecoherence2004}, and in describing other quantum correlations besides entanglement such as discord \cite{xdiscord2010}.
In Theorem \ref{le:xstate} we study the inertia of $\rho^\G$ by quantifying the number of negative eigenvalues of $\rho^\G$.


\begin{theorem}
\label{le:xstate} 
(i) If $\rho$ is a $2\times n$ X-state, then $\rho^\G$ has at most $\lfloor\frac{n}{2}\rfloor$ negative eigenvalues. Furthermore, there exist $2\times n$ X-states whose partial transpose has exact $k$ negative eigenvalues, where $1\leq k\leq \lfloor\frac{n}{2}\rfloor$.
\end{theorem}

One can refer to Appendix \ref{sec:proof3} for the proof of Theorem \ref{le:xstate}. In \cite{2x3mes2017} authors considered the problems of maximizing the entanglement negativity of qubit-qutrit X-states. For this purpose they derived that there is at most one negative eigenvalue of $\rho^\Gamma$ if $\rho$ is a qubit-qutrit X-state. We generalize their result to $2\times n$ X-states here. In the proof of Theorem \ref{le:xstate} we formulate expressions for the eigenvalues of $\rho^\G$. Therefore, using the expressions of those negative eigenvalues, one can determine the inertia of $\rho^\G$, and compute the negativity of $2\times n$ X-state.

Finally, since the transpose is a typical positive but not completely positive map, it enables PPT criterion to detect entanglement. In \cite{inverseew2018} authors considered the question of how exactly the partial transpose map can transform the eigenvalues of $\rho$. In specific, for which ordered list $\l_1\geq \l_2\geq \cdots\geq \l_{mn}\in\mathbb{R}$ does there exist an $m\times n$ state $\rho$ such that $\rho^\G$ has eigenvalues $\l_1,\cdots,\l_{mn}$? Our main result Theorem \ref{le:numcN2n-1} answers this question when $m=2$ in terms of how many positive and negative values among the list $\l_1\geq \l_2\geq \cdots\geq \l_{2n}$.

\section{Conclusions}
\label{sec:con}

In this paper, we investigated the inertia of EW with a focus on the partial transpose of NPT state in the qubit-qudit system. 
For general EWs, we obtained the lower and upper bounds on the number of negative (positive) eigenvalues. 
We also completely determined the inertias of two-qubit EWs.
Then we focused on the EWs constructed by the partial transposition of NPT states.
On the one hand, we proposed an effective method to generate more inertias from a given inertia by adding appropriate product states into the density matrix of a given state.
On the other hand, we introduced an essential tool to study the inertia. That is the inertia of a Hermitian operator is invariant under $\slocc$. Using such a tool we presented a sufficient and necessary condition for a sequence belonging to the inertia set $\cN_{2,n}$. 
Combining the above two results we exhausted all inertias in the inertia set $\cN_{2,n}$ for any $n\geq 2$. 
Moreover, it was indicated that $(1,mn-4,3)\in\cN_{m,n}$, and the least number of positive eigenvalues is three.
These results led us to understand the partial transpose of NPT state better. 
Finally we connected our results to other problems in quantum information theory. Two important applications of our results were indicated. 
First, we proposed a separability criterion in terms of $r(\rho^\G)$. 
Second, using the essential property that the inertia is invariant under $\slocc$ equivalence, we introduced a classification of NPT states.
To sum up, our results carried out the first step to study the inertia of EW, and can be used to study other unexplored problems.


There are some interesting problems for further study. First, it is natural to ask how many distinct inertias in the inertia set $\cN_{m,n}$ for $m,n\geq 3$. We conjecture that $\cN_{m,n}$ has exact $(m-1)^2(n-1)^2$ distinct inertias. As we mentioned above, the inertia of $\rho^\G$ is closely related to the $1$-distillability of NPT state $\rho$. The characterization of $\cN_{m,n}$ for $m,n\geq 3$ is helpful to understand the distillability problem. Second, we may extend the study to the inertia of non-decomposable EW. The partial transpose of NPT state is a decomposable EW which cannot detect PPTE states. Thus, the inertias of non-decomposable EWs may provide powerful separability criteria to identify PPTE states. Third, we may generalize the bi-rank $(r(\rho),r(\rho^\G))$ to the bi-inertia $(\ine(\rho),\ine(\rho^\G))$. In virtue of the bi-inertias we can determine the partial transposes of which states share the same inertia. It allows us to further understand the relation between $\rho$ and $\rho^\G$. Finally, it is interesting to dig more applications of the tools introduced in this paper.

\begin{acknowledgments}
We deeply appreciate the anonymous referees for their careful work and valuable suggestions. We thank Bang-hai Wang and You Zhou for their helpful discussion. This work was supported by the NNSF of China (Grant Nos. 11871089, 11947241), and the Fundamental Research Funds for the Central Universities (Grant No. ZG216S2005). 
\end{acknowledgments}

\appendix

\section{Proofs of results in Sec. \ref{sec:restrict}}
\label{sec:proofsec2}

First of all, we prepare to show the proof of Lemma \ref{le:inertia}. For this purpose we need two essential results. The first one is a well-known conclusion on the existence of product vectors in some subspace.

\begin{lemma}\cite[Proposition 6.]{pptsprop2013}
\label{le:dimprod}
Suppose $\cH^{AB}\cong\mathbb{C}^{m}\ox\mathbb{C}^n$ is a bipartite Hilbert space.
Any subspace of $\cH^{AB}$ with dimension greater than $(m-1)(n-1)$ must contain at least one product vector. Furthermore, any subspace of $\cH^{AB}$ with dimension greater than $(m-1)(n-1)+1$ contains infinitely many product vectors.
\end{lemma}

The second one is on the dimension of some subspace spanned by product vectors.

\begin{lemma}
\label{le:prodspan}
Suppose $\cH^{AB}\cong\mathbb{C}^{m}\ox\mathbb{C}^n$ is a bipartite Hilbert space.
If $V$ is an $(mn-1)$-dimensional bipartite subspace of $\cH^{AB}$, then $V$ is spanned by product vectors. If $V$ is an $(mn-2)$-dimensional bipartite subspace of $\cH^{AB}$, then $V$ may be not spanned by product vectors.
\end{lemma}

\begin{proof}
Let $V$ be the subspace spanned by the linearly independent vectors $\ket{\a_1},\ket{\a_2},...,\ket{\a_{mn-1}}$ in $\bbC^m\otimes\bbC^n$. It is known that the 3-tensor 
\begin{eqnarray}
\label{eq:mn(mn-2)=1}
\ket{\ps}=
\sum^{mn-1}_{j=1}\ket{\a_j}\ket{j}	
\end{eqnarray}
has tensor rank $(mn-1)$ \cite{3-tensorrank1983}. 
That is,
\begin{eqnarray}
\label{eq:mn(mn-2)=2}
\ket{\ps}=	
\sum^{mn-1}_{j=1}\ket{a_j,b_j,c_j},
\end{eqnarray}
where $\ket{c_j}$'s are vectors in the space 

\begin{eqnarray}
\lin\{\ket{1},\ket{2},...,\ket{mn-1}\}.	
\end{eqnarray}
Comparing \eqref{eq:mn(mn-2)=1} and \eqref{eq:mn(mn-2)=2}, we obtain that
$\ket{c_j}$'s are linearly independent. Hence,
\begin{eqnarray}
V
&=&
\lin\{\ket{\a_1},...,\ket{\a_{mn-1}}\}
\\\notag\\
&=&	
\lin\{\ket{a_1,b_1},...,\ket{a_{mn-1},b_{mn-1}}\}.
\end{eqnarray}
It follows that $V$ is spanned by product vectors.	
To prove the second claim, it suffices to construct an example in $\mathbb{C}^2\ox\mathbb{C}^2$. For example, let $V$ be the 2-dimensional subspace spanned by $\ket{00}$ and $\ket{01}+\ket{10}$. One can verify that $V$ has exact one product vector $\ket{00}$ up to a coefficient. This completes the proof. 
\end{proof}

Then we are ready to prove Lemma \ref{le:inertia}.

\textbf{Proof of Lemma \ref{le:inertia}.}
First, we write $W$ in spectral decomposition as
\begin{equation}
W=\sum_{i=1}^{mn}p_{i}\proj{a_{i}}\in \cB(\mathbb{C}_{m}\otimes\mathbb{C}_{n}).
\end{equation}
Suppose the inertia of $W$ is $(q,r,mn-q-r)$. Without loss of generality, we may assume that $p_i\le p_{i+1}$ for any $i$. Then the eigenvalues $p_i<0$ for $i\in[1,q]$, $p_j=0$ for $j\in[q+1,q+r]$, and $p_k>0$ for $k\in[r+q+1,mn]$. Since $W$ is an EW, we have $q\ge1$ and $mn-q-r\ge1$.

(i) The assertion follows from the definition of EW. 

(ii)  If $q\geq (m-1)(n-1)+1$, from Lemma \ref{le:dimprod} there is a product state $\ket{a,b}$ in the subspace spanned by $\{\ket{a_{1}},\ket{a_{2}},...,\ket{a_{q}}\}$. So $\bra{a,b}W\ket{a,b}<0$. It is a contradiction with the definition of EW. Hence, $q\le (m-1)(n-1)$. The decomposable EW containing exactly $(m-1)(n-1)$ negative eigenvalues has been constructed in \cite{Johnston2013Non}.

(iii) We prove the assertion by contradiction. Suppose $mn-q-r=1$. Lemma \ref{le:prodspan} implies that there is a product vector $\ket{a,b}$ orthogonal to $\ket{a_{mn}}$, and non-orthogonal to $\ket{a_1}$. We have
\begin{eqnarray}
0&\le &\bra{a,b}W\ket{a,b}	
\le 
p_1
\bra{a,b}
(\proj{a_1})
\ket{a,b}
<0.
\end{eqnarray}
We obtain a contradiction. Therefore, $mn-q-r>1$, namely that $W$ has at least two positive eigenvalues.

This completes the proof.
\qed

\section{Proofs of results in Sec. \ref{sec:inebiew}.}
\label{sec:proof2}

First we show the proof of Lemma \ref{le:rho+xid} as follows.

\textbf{Proof of Lemma \ref{le:rho+xid}.}
(i) Since $\ine(\r^\G)=(a,b,c)$, we may assume the spectral decomposition as
\begin{eqnarray}
\r^\G
&=&\sum^a_{i=1}\l_i\proj{\a_i}
+0\cdot \sum^b_{j=1}\proj{\b_j}
\notag\\\notag\\
&+&\sum^c_{k=1}\m_k\proj{\g_k}
,	
\end{eqnarray}
where $\l_i<0$ and $\m_k>0$. We choose $x>0$ such that
\begin{eqnarray}
x+\max_i\l_i<0.	
\end{eqnarray}
Therefore,
\begin{eqnarray}
\s^\G &=&\r^\G+xI_{mn} 
\notag\\\notag\\
&=&\sum^a_{i=1}(x+\l_i)\proj{\a_i}
+x\cdot \sum^b_{j=1}\proj{\b_j}
\notag\\\notag\\
&+&\sum^c_{k=1}(x+\m_k)\proj{\g_k}
.	
\end{eqnarray}
It follows that
$$
\ine(\s^\G)=(a,0,b+c).
$$
Furthermore, if $\rho$ is an NPT (separable) state, then $\s$ is also an NPT (separable) state. If $\rho$ is a PPTE state, there is an EW $W$ such that $\tr(W\rho)<0$. So we can choose $x>0$ such that $x+\max\limits_{i}\l_i<0$, and $\tr(W(\rho+x I))<0$. Thus $\s$ is also a PPTE state.

(ii) If $(a_1,b_1,c_1)\in\cN_{m_1,n_1}(\cP_{m_1,n_1},\cS_{m_1,n_1})$, it follows from (i) that 
$$(a_1,0,b_1+c_1)\in\cN_{m_1,n_1}(\cP_{m_2,n_2},\cS_{m_2,n_2}).$$
Suppose $\rho$ is an $m_1\times n_1$ NPT (PPTE, separable) state with $\ine(\rho^\G)=(a_1,0,b_1+c_1)$. Using the spectral decomposition we write $\rho^\G$ as
\beq
\label{eq:inerelationmn-1}
\rho^\G=\sum_{j=1}^{a_1}\l_j\proj{\psi_j}+\sum_{j=a_1+1}^{b_1+c_1}\m_j\proj{\psi_j},
\eeq
where $\l_j<0,\m_j>0$, and $\{\ket{\psi_j}\}_{j=1}^{m_1n_1}$ is an orthonormal basis of $\mathbb{C}^{m_1}\ox\mathbb{C}^{n_1}$. By adding proper zero rows and columns in the origional density matrix of $\rho$ we construct an $m_2\times n_2$ NPT (PPTE, separable) state $\tilde{\rho}$, and
$$\ine(\tilde{\rho}^\G)=(a_1,m_2n_2-m_1n_1,b_1+c_1).$$
We again write $\tilde{\rho}^\G$ in spectral decomposition as
\beq
\label{eq:inerelationmn-2}
\bal
\tilde{\rho}^\G&=\sum_{j=1}^{a_1}\l_j\proj{\tilde{\psi_j}}+\sum_{j=a_1+1}^{b_1+c_1}\m_j\proj{\tilde{\psi_j}}\\
&+0\cdot \sum_{j=1}^{m_2n_2-m_1n_1}\proj{\phi_j},
\eal
\eeq
where 
$\{\ket{\phi_j}\}_{j=1}^{m_2n_2-m_1n_1}$ is the set of product vectors $\ket{p,q}$ with either $m_1<p\leq m_2$ or $n_1< q\leq n_2$. Let
\beq
\label{eq:inerelationmn-3}
\sigma^\G:=\tilde{\rho}^\G+\sum_{x=1}^l\proj{j_x,k_x},
\eeq
where 
$\{\ket{j_x,k_x}\}_{x=1}^l\subseteq \{\ket{\phi_j}\}_{j=1}^{m_2n_2-m_1n_1}$ is a subset of any $l$ orthonormal product vectors. Thus, 
$$\ine(\sigma^\G)=(a_1,m_2n_2-m_1n_1-l,b_1+c_1+l)$$
for any $0\leq l\leq m_2n_2-m_1n_1$. Since the state $\s$ is obtained by adding product states into the density matrix of $\tilde{\rho}$, it follows that $\s$ is also an $m_2\times n_2$ NPT (PPTE, separable) state. 

This completes the proof.
\qed

The basic idea of Lemma \ref{le:rho+xid} is to add a convex combination of linearly independent product states into the original density matrix. Using the similar idea we can also determine the inertia sets $\cS_{m,n}$ and $\cP_{m,n}$ as by-products.

\begin{corollary}
Suppose $\r$ is an $m\times n$ PPT state. Then 
\begin{eqnarray}
\label{eq:inertia=sep}
\ine (\rho^\G) =\big(0,mn-r(\r^\G),r(\r^\G)\big).
\end{eqnarray}

(i) If $\r$ is separable then any given integer $r(\r^\G)\in[1,mn]$ exists.

(ii) If $\r$ is a PPTE state, and 
$$
k:=\max\{\n_0|(0,\n_0,\n_+)\in\cP_{m,n}\},
$$
then  $(0,mn-r\big(\r^\G),r(\r^\G)\big)\in\cP_{m,n}$ for any given integer $r(\r^\G)\in[mn-k,mn]$.
\end{corollary}

\begin{proof}
Since $\r$ is a PPT state, $\r^\G$ has no negative eigenvalue. Hence, Eq. \eqref{eq:inertia=sep} holds. 

(i) It suffices to construct specific examples to prove this assertion. Let
\begin{eqnarray}
\r=
\sum^m_{i=1}\sum^n_{j=1}	
c_{ij}\proj{ij},
\end{eqnarray}	
where $c_{ij}$'s are non-negative real numbers, and exact $p\in[1,mn]$ elements of $\{c_{ij}\}$ are positive. It follows that $r(\r)=r(\r^\G)=p$ for any $p\in[1,mn]$.

(ii) Suppose $\rho$ is a PPTE state which satisfies $r(\r^\G)=mn-k$. We write $\rho^\G$ in spectral decomposition as
\beq
\label{eq:ppteine}
\r^\G=\sum_{j=1}^{mn-k} \l_j\proj{\psi_j},
\eeq 
where $\l_j$'s are positive.
It follows that $\cR(\r^\G)=\lin\{\ket{\psi_j}\}_{j=1}^{mn-k}$. Thus we can assume the kernel of $\r^\G$ is spanned by $k$ linearly independent product vectors, i.e., 
$$\cK(\r^\G)=\lin\{\ket{a_j,b_j}\}_{j=1}^k.$$ Hence, for any $1\leq p\leq k$ we define
\beq
\label{eq:ppteine-1}
\s_p^\G:=\sum_{j=1}^{mn-k} \l_j\proj{\psi_j}+\sum_{j=1}^p\proj{a_j,b_j}.
\eeq 
One can verify $\s_p$ is a PPTE state which satisfies $r(\s_p^\G)=mn-k+p, ~\forall 1\leq p\leq k$.

This completes the proof.
\end{proof}

Second we provide the proof of Lemma \ref{le:mxnkerprod} as follows.

\textbf{Proof of Lemma \ref{le:mxnkerprod}.}
Using the spectral decomposition we can write $\rho^\G$ as
\beq
\label{eq:mxnkerprod-1}
\rho^\G=-\sum_{j=1}^k\proj{v_j}+\sum_{j=k+1}^{mn-d}\proj{v_j},
\eeq
where $\{\ket{v_j}\}_{j=1}^{mn-d}$ are pairwisely orthogonal. Assume $\cK(\r^\G)=\lin\{\ket{u_1},\cdots,\ket{u_d}\}$. It follows from Lemma \ref{le:dimprod} that any subspace of $\mathbb{C}^m\ox\mathbb{C}^n$ whose dimension is $(m-1)(n-1)+1$ contains at least one product vector. Thus, there exist proper coefficients such that 
\beq
\label{eq:mxnkerprod-2}
\ket{a_1,b_1}=\sum_{i=1}^k x_i\ket{v_i}+\sum_{j=1}^{(m-1)(n-1)+1-k}y_j\ket{u_j}.
\eeq
Since $\rho^\G$ is an EW, it follows that $\bra{a_1,b_1}\rho^\G\ket{a_1,b_1}\geq 0$. Thus we conclude that $x_i=0,\forall 1\leq i\leq k$ in Eq. \eqref{eq:mxnkerprod-2}. That is
\beq
\label{eq:mxnkerprod-2.1}
\ket{a_1,b_1}=\sum_{j=1}^{(m-1)(n-1)+1-k}y_j\ket{u_j}\in\cK(\rho^\G).
\eeq
Up to a permutation of $\{\ket{u_j}\}_{j=1}^{(m-1)(n-1)+1-k}$, we can assume $y_1\neq 0$. In the same way, there exist proper coefficients such that 
\beq
\label{eq:mxnkerprod-3}
\ket{a_2,b_2}=\sum_{j=2}^{(m-1)(n-1)+2-k}y_j\ket{u_j}\in\cK(\rho^\G).
\eeq
Similarly we assume that $y_2\neq 0$. Repeating this process we obtain that there are at least $l$ linearly independent product vectors in $\cK(\r^\G)$. Moreover, we conclude that 
\beq
\label{eq:kerprod-4}
\bal
\cK(\rho)&=\lin\{\ket{u_1},\cdots,\ket{u_d}\}\\
&=\lin\{\ket{a_1,b_1},\cdots,\ket{a_l,b_l},\ket{u_{l+1}},\cdots,\ket{u_d}\}.
\eal
\eeq

This completes the proof.
\qed

Third we show the proof of Theorem \ref{le:n-1negative} as follows.

\textbf{Proof of Theorem \ref{le:n-1negative}.}
(i) We first prove the "Only if" part. It is equivalent to prove the claim that if the sequence $(a,b,c)$ satisfies that either $(a,b-2,c)\in\cN_{2,n}$ or $(a,b-1,c-1)\in\cN_{2,n}$, then $(a,b,c)\in\cN_{2,n}$. If $(a,b-2,c)\in\cN_{2,n-1}$, then $(a,b,c)\in\cN_{2,n}$ naturally. If $(a,b-1,c-1)\in\cN_{2,n-1}$, then $(a,b+1,c-1)\in\cN_{2,n}$ naturally. Suppose $\sigma$ is a $2\times n$ state, and $\ine(\sigma^\G)=(a,b+1,c-1)$. Since $a+b+1>n$, from Lemma \ref{le:mxnkerprod} there is a product vector in $\cK(\sigma^\G)$, namely $\ket{e,f}$. Let 
$$\tilde{\sigma}:=\sigma+\proj{e^*,f}.$$ It follows that $\ine(\tilde{\sigma}^\G)=(a,b,c)$, and thus $(a,b,c)\in\cN_{2,n}$. So the "Only if" part holds.

Second we prove the "If" part by contradiction. Assume that there is a $2\times n$ NPT state $\rho$ such that $\ine(\rho^\G)=(a,b,c)$. Since $a+b>n-1$, it follows from Lemma \ref{le:mxnkerprod} that there is a product vector in $\cK(\rho^\Gamma)$. Thus we can assume $\ket{0,0}\in\cK(\rho^\G)$ up to SLOCC equivalence. Also, we obtain that $\ket{0,0}\in\cK(\rho)$. Hence, the matrix of $\rho^\G$ is as follows. 
\beq
\label{eq:n-1neg-1}
\rho^\G=
\left[
\begin{array}{c|c}
M_{11} & M_{12}  \\ \hline 
M_{21} & M_{22}
\end{array}
\right],
\eeq
where
\beq
\label{eq:n-1neg-2}
\bal
M_{11}&=
\bma
0 & 0 & 0 & \cdots & 0 \\
0 & \rho_{22} & \rho_{23} & \cdots & \rho_{2n}\\
0 & \rho_{32} & \rho_{33} & \cdots & \rho_{3n}\\
\vdots & \vdots & \vdots & \vdots & \vdots \\
0 & \rho_{n2} & \rho_{n3} & \cdots & \rho_{nn}\\
\ema,\\
M_{12}&=
\bma
0 & 0 & 0 & \cdots & 0 \\
0 & \rho_{2(n+2)} & \rho_{2(n+3)} & \cdots & \rho_{2(2n)}\\
0 & \rho_{3(n+2)} & \rho_{3(n+3)} & \cdots & \rho_{3(2n)}\\
\vdots & \vdots & \vdots & \vdots & \vdots \\
0 & \rho_{n(n+2)} & \rho_{n(n+3)} & \cdots & \rho_{n(2n)}\\
\ema,\\
M_{21}&=
\bma
0 & 0 & 0 & \cdots & 0 \\
0 & \rho_{(n+2)2} & \rho_{(n+2)3} & \cdots & \rho_{(n+2)n}\\
0 & \rho_{(n+3)2} & \rho_{(n+3)3} & \cdots & \rho_{(n+3)n)}\\
\vdots & \vdots & \vdots & \vdots & \vdots \\
0 & \rho_{(2n)2} & \rho_{(2n)3} & \cdots & \rho_{(2n)n}\\
\ema,\\
M_{22}&=
\bma
\rho_{(n+1)(n+1)} & \rho_{(n+1)(n+2)} & \cdots & \rho_{(n+1)(2n)} \\
\rho_{(n+2)(n+1)} & \rho_{(n+2)(n+2)} & \cdots & \rho_{(n+2)(2n)}\\
\rho_{(n+3)(n+1)} & \rho_{(n+3)(n+2)} & \cdots & \rho_{(n+3)(2n))}\\
\vdots & \vdots & \vdots & \vdots \\
\rho_{(2n)(n+1)} & \rho_{(2n)(n+2)}  & \cdots & \rho_{(2n)(2n)}\\
\ema.\\
\eal
\eeq

If $\rho_{(n+1)(n+1)}=0$, since $\rho$ is positive semidefinite, we conclude that
$$
\rho_{(n+1)j}=\rho_{j(n+1)}=0, \quad \forall n+2\leq j\leq 2n.
$$ 
It implies that $\ket{1,0}\in\cK(\rho^\G)$, and $\ket{1,0}\in\cK(\rho)$. Thus, $\rho$ is indeed a $2\times (n-1)$ state up to a local projector. It implies that $(a,b-2,c)\in\cN_{2,n-1}$. It contradicts with the condition (i.a) that $(a,b-2,c)\not\in\cN_{2,n-1}$. 
Next, we consider $\rho_{(n+1)(n+1)}>0$.
Using a locally invertible operator $I_2\ox V$, where $V$ is an $n\times n$ invertible matrix, we obtain that 
\beq
\label{eq:n-1neg-3}
(I\ox V)\rho^\G(I\ox V^\dg)=
\left[
\begin{array}{c|c}
M_{11} & M_{12}  \\ \hline 
M_{21} & M_{22}'
\end{array}
\right],
\eeq
where
\beq
\label{eq:n-1neg-3.1}
M_{22}'=
\bma
\rho_{(n+1)(n+1)} & 0 & \cdots & 0 \\
0 & \rho_{(n+2)(n+2)}' & \cdots & \rho_{(n+2)(2n)}'\\
0 & \rho_{(n+3)(n+2)}' & \cdots & \rho_{(n+3)(2n))}'\\
\vdots & \vdots & \vdots & \vdots \\
0 & \rho_{(2n)(n+2)}' & \cdots & \rho_{(2n)(2n)}'\\
\ema.
\eeq
Since the inertia is invariant under invertible operations by Sylvester Theorem,
we conclude that the inertia of 
\beq
\label{eq:n-1neg-4}
\sigma^\G=
\bma
\rho_{22} & \cdots & \rho_{2n} & \rho_{2(n+2)} & \cdots & \rho_{2(2n)} \\
\rho_{32} & \cdots & \rho_{3n} & \rho_{3(n+2)} & \cdots & \rho_{3(2n)} \\
\vdots & \vdots & \vdots & \vdots & \vdots & \vdots \\
\rho_{n2} & \cdots & \rho_{nn} & \rho_{n(n+2)} & \cdots & \rho_{n(2n)} \\
\rho_{(n+2)2} & \cdots & \rho_{(n+2)n} & \rho_{(n+2)(n+2)}' & \cdots & \rho'_{(n+2)(2n)} \\
\rho_{(n+3)2} & \cdots & \rho_{(n+3)n} & \rho_{(n+3)(n+2)}' & \cdots & \rho'_{(n+3)(2n)} \\
\vdots & \vdots & \vdots & \vdots & \vdots & \vdots \\
\rho_{(2n)2} & \cdots & \rho_{(2n)n} & \rho'_{(2n)(n+2)} & \cdots & \rho'_{(2n)(2n)}
\ema
\eeq
is $(a,b-1,c-1)$. Since $\sigma$ is a $2\times (n-1)$ state, it follows that $(a,b-1,c-1)\in\cN_{2,n-1}$. It contradicts with the condition (i.b) that $(a,b-1,c-1)\not\in\cN_{2,n-1}$. Therefore, we conclude that $(a,b,c)\not\in\cN_{2,n}$. So the "If" part holds. 

(ii) Since $\ine(\rho^\G)=(a,b,c)$ with $a+b>n-1$, it follows from Lemma \ref{le:mxnkerprod} that there is a product vector in $\cK(\rho^\G)$. Up to SLOCC equivalence, we may assume $\ket{0,0}\in\cK(\rho^\G)$. Thus, one can similarly write $\rho^\G$ in the form as Eq. \eqref{eq:n-1neg-1}. For the entry $\rho_{(n+1)(n+1)}$ in $M_{22}$ given by \eqref{eq:n-1neg-2}, if it is positive, then we can transform $\rho$ into a $2\times (n-1)$ state $\sigma$ by using a locally invertible operation, namely Eq. \eqref{eq:n-1neg-3}. Similarly, $\sigma^\G$ expressed by Eq. \eqref{eq:n-1neg-4} has inertia $(a,b-1,c-1)$. It contradicts with $(a,b-1,c-1)\not\in\cN_{2,n-1}$. Therefore, we conclude that $\rho_{(n+1)(n+1)}=0$. It implies $\ket{1,0}\in\cK(\rho^\G)$. Thus, $\rho$ is indeed a $2\times(n-1)$ state up to a local projector, and $(a,b-2,c)\in\cN_{2,n-1}$. So assertion (ii) holds.

This completes the proof.
\qed

Fourth, we show the proof of Corollary \ref{th:2x3ine} as follows.

\textbf{Proof of Corollary \ref{th:2x3ine}.}
(i) It follows from Lemma \ref{le:pure+upper} (iii) that such a state $\rho$ whose partial transpose contains exact $(n-1)$ negative eigenvalues exists. Then we show $\ine(\rho^\G)$ can only be $(n-1,0,n+1)$. First for any $k>0$ the sequence $(n-1,k,n+1-k)$ satisfies the condition $n-1+k>n-1$. Second it follows from Lemma \ref{le:inertia} (ii) that the partial transpose of any $2\times (n-1)$ NPT state has at most $(n-2)$ negative eigenvalues. Hence, we conclude that for any $k>0$,
\beq
\label{eq:n-1ine-1}
\bal
&(n-1,k-2,n+1-k)\not\in\cN_{2,n-1},\\
&(n-1,k-1,n-k)\not\in\cN_{2,n-1}.
\eal
\eeq
Note that if $k-2<0$ or $n+1-k<0$, then such inertia $(n-1,k-2,n+1-k)$ naturally does not exist. Similarly if $k-1<0$ or $n-k<0$, then such inertia $(n-1,k-1,n-k)$ does not exist.
Therefore, it follows from Theorem \ref{le:n-1negative} (i) that $(n-1,k,n+1-k)\not\in\cN_{2,n}$ for any $k>0$.

(ii) It follows from Lemma \ref{le:inertia} that the number of negative eigenvalues of $\rho^\G$ is either one or two, and the number of positive eigenvalues of $\rho^\G$ lies in $[2,5]$. Thus, $\ine(\rho^\G)$ can only be the following seven sequences:
\beq
\label{eq:2x3pseq}
\bal
&(1,3,2), (1,2,3), (1,1,4), (1,0,5),\\
&(2,2,2), (2,1,3), (2,0,4).
\eal
\eeq
First we construct concrete examples to show the four inertias in Eq. \eqref{eq:2x3ptineexist} exist.
\beq
\label{eq:2x3existseqex-1}
\bal
\r_1&=(\ket{00}+\ket{11})(\bra{00}+\bra{11}),\\
\r_2&=(\ket{00}+\ket{11})(\bra{00}+\bra{11})+\proj{02},\\
\r_3&=(\ket{00}+\ket{11})(\bra{00}+\bra{11})+\frac{1}{10}I_6,\\
\r_4&=(\ket{00}+\ket{11})(\bra{00}+\bra{11})\\
    &+(\ket{01}+\ket{12})(\bra{01}+\bra{12}).
\eal
\eeq
One can verify
\beq
\label{eq:2x3existseqex-2}
\bal
\ine(\r_1)&=(1,2,3), \quad \ine(\r_2)&=(1,1,4),\\
\ine(\r_3)&=(1,0,5),\quad \ine(\r_4)&=(2,0,4).
\eal
\eeq

Second we exclude other three sequences in \eqref{eq:2x3pseq}. It follows from Theorem \ref{cr:twoqubit} that $\cN_{2,2}=\{(1,0,3)\}$. Then from Theorem \ref{le:n-1negative} (i) the three sequences $(1,3,2),(2,2,2),(2,1,3)$ do not belong to $\cN_{2,3}$.


(iii) We prove it by contradiction. Assume $r(\rho_B)=k(\neq j+1)$. If $k< j+1$, then $r(\rho_{AB}^{\G})\leq 2k<2(j+1)$. It contradicts with $\ine(\rho_{AB}^\G)=(j,2(n-1-j),j+2)$. If $k> j+1\geq 2$, then we can assume $\ine(\rho_{AB}^\G)=(j,2(k-1-j),j+2)$ by taking $\rho_{AB}$ as a $2\times k$ state. Since $j<k-1$, from \eqref{eq:numcN2n-3} we obtain $(j,2(k-2-j),j+2)\in\cN_{2,k-1}$. Moreover, for any inertia $(a,b,c)\in\cN_{2,n},~\forall n\geq 2$, an observation from \eqref{eq:numcN2n-3} is that $c-a\geq 2$. Hence, $(j,2(k-2-j)+1,j+1)\not\in\cN_{2,k-1}$. Straightforward calculation yields that $j+2(k-1-j)>k-1$ from $k> j+1$. It follows from Theorem \ref{le:n-1negative} (ii) that $\rho_{AB}$ can be regarded as $2\times (k-1)$ state. It implies $r(\rho_B)\leq k-1$. It contradicts with the assumption $r(\rho_B)=k$. Therefore, we conclude that $r(\rho_B)=j+1$.

This completes the proof.
\qed

Finally we provide the proof of Theorem \ref{le:numcN2n-1} as follows.

\textbf{Proof of Theorem \ref{le:numcN2n-1}.}
First we show the $(n-1)^2$ sequences in \eqref{eq:numcN2n-3} belong to $\cN_{2,n}$.
It follows from Corollary \ref{th:2x3ine} (i) that 
$$
(j-1,0,j+1)\in\cN_{2,j},~\forall 2\leq j\leq n.
$$
Then from Lemma \ref{le:rho+xid} (ii) we conclude that $\forall 2\leq j\leq n$,
\beq
\label{eq:numcN2n-2}
(j-1,2(n-j)-l,j+1+l)\in\cN_{2,n},~\forall 0\leq l\leq 2(n-j).
\eeq 
Thus the number of distinct inertias in $\cN_{2,n}$ is at least
$$
\sum_{j=2}^n \big(2(n-j)+1\big)=(n-1)^2.
$$

Second we show except the $(n-1)^2$ sequences in \eqref{eq:numcN2n-3} there is no other inertia in $\cN_{2,n}$. We prove this claim using mathematical induction. First, it follows from Theorem \ref{cr:twoqubit} and Corollary \ref{th:2x3ine} (ii) that this claim holds for $n=2,3$. Assume $\abs{\cN_{2,n}}=(n-1)^2$ holds for $n=k$. Next, we need to show $\abs{\cN_{2,n}}=(n-1)^2$ holds for $n=k+1$. From \eqref{eq:numcN2n-2}, it is equivalent to prove that 
\begin{eqnarray}
\label{eq:(j-1,b,c)}
(j-1,b,c)\not\in\cN_{2,k+1},	
\end{eqnarray}
for any $2\leq j\leq k+1$, 
where $b>2(k+1-j)$ and $b+c=2(k+1)+1-j$. Straightforward computation yields that 
$
j-1+b>2k+1-j\geq k.
$
Thus we can apply Theorem \ref{le:n-1negative} (i) to prove \eqref{eq:(j-1,b,c)}. This is equivalent to prove that
\begin{eqnarray}
\label{eq:(j-1,b-2,c)}
(j-1,b-2,c),\quad
(j-1,b-1,c-1)\not\in\cN_{2,k},	
\end{eqnarray}
for any $2\leq j\leq k+1$, 
where $b>2(k+1-j)$ and $b+c=2(k+1)+1-j$. 

According to the induction hypothesis we obtain 
\beq
\label{eq:induction-1}
\bal
\cN_{2,k}&=\bigg\{(j-1,2(k-j)-l,j+1+l)\bigg| \\
&\forall 0\leq l\leq 2(k-j), ~\forall 2\leq j\leq k\bigg\}.
\eal
\eeq
Thus for any $2\leq j\leq k$, $(j-1,b-2,c)\in\cN_{2,k}$ if and only if $0\leq b-2\leq 2(k-j)$. This is a contradiction with the condition $b>2(k+1-j)$ below \eqref{eq:(j-1,b-2,c)}. Similarly for any $2\leq j\leq k$, using Eq. \eqref{eq:induction-1} we obtain that $(j-1,b-1,c-1)\in\cN_{2,k}$ if and only if $0\leq b-1\leq 2(k-j)$. So we obtain the same contradiction. We have proven \eqref{eq:(j-1,b-2,c)} for $2\le j\le k$.

It remains to prove \eqref{eq:(j-1,b-2,c)}  
for $j=k+1$. It follows from Lemma \ref{le:inertia} (ii) that the partial transpose of any $2\times k$ NPT state has at most $k-1$ negative eigenvalues. Thus, both $(j-1,b-2,c)$ and $(j-1,b-1,c-1)$ do not belong to $\cN_{2,k}$ if $j=k+1$. We have proven \eqref{eq:(j-1,b-2,c)} for $j=k+1$. Combining with the last paragraph, we have proven \eqref{eq:(j-1,b-2,c)}. The equivalence of \eqref{eq:(j-1,b,c)} and \eqref{eq:(j-1,b-2,c)} implies that
$\abs{\cN_{2,n}}=(n-1)^2$ holds for $n=k+1$.

To sum up, according to mathematical induction we conclude that $\abs{\cN_{2,n}}=(n-1)^2$ for any $n\geq 2$. This completes the proof.
\qed

\section{Proofs of results in Sec. \ref{sec:app}.}
\label{sec:proof3}

First we show the following results.

\begin{lemma}
\label{le:sloccequiv}
(i) If two $n$-partite mixed states of system $A_1,...,A_n$ are SLOCC equivalent, then their partial transposes with respect to any $k$-partite subsystem $A_{j_1},....,A_{j_k}$ are SLOCC equivalent.

(ii) Suppose $\rho_{AB}$ and $\sigma_{AB}$ are both $2\times n$ NPT states of system $A,B$. If $\rho_{AB}^\G$ and $\sigma_{AB}^\G$ have different inertias, then the partial transposes of $\rho^{\ox N}$ and $\sigma^{\ox N}$ still have different inertias for any $N$ copies.
\end{lemma}

\begin{proof}
(i) Suppose $\rho$ and $\sigma$ are two $n$-partite mixed states of system $A_1,...,A_n$, and they are SLOCC equivalent. Let $\rho^\G$ and $\sigma^\G$ be the partial transposes of $\rho$ and $\sigma$ respectively, with respect to first $k$-partite subsystem $A_1,\cdots, A_k$. Up to a permutation of subsystems, it suffices to show that $\rho^\G$ and $\sigma^\G$ are SLOCC equivalent. By Definition \ref{df:equivalence} there is a locally invertible operator 
$$
X=V_1\ox V_2\ox \cdots\ox V_n
$$ 
such that $X\rho X^\dg=\sigma$. Let
$$
X^\G:=V_1^T\ox \cdots\ox V_k^T\ox V_{k+1}\ox\cdots\ox V_n.
$$ 
One can verify $(X^\G)^\dg\rho^\G X^\G=\sigma^\G$. Therefore, $\rho^\G$ and $\sigma^\G$ are SLOCC equivalent.

(ii) Denote $\ine(\rho_{AB}^\G)=(a_1,b_1,c_1)$ and $\ine(\sigma_{AB}^\G)=(a_2,b_2,c_2)$. It follows that
$$
(\rho_{AB}^{\ox N})^\G=(\rho_{AB}^\G)^{\ox N}, \quad   (\sigma_{AB}^{\ox N})^\G=(\sigma_{AB}^\G)^{\ox N}.
$$
Straightforward calculation yields that
\beq
\label{eq:ncopyines-1}
\bal
\n_-\big((\rho_{AB}^{\ox N})^\G\big)&=\sum_{k-odd} {N\choose k} a_1^kc_1^{N-k}\\
&=\frac{(a_1+c_1)^N-(a_1-c_1)^N}{2},\\
\n_+\big((\rho_{AB}^{\ox N})^\G\big)&=\sum_{k-even} {N\choose k} a_1^kc_1^{N-k}\\
&=\frac{(a_1+c_1)^N+(a_1-c_1)^N}{2}.\\
\eal
\eeq
Similarly we obtain
\beq
\label{eq:ncopyines-2}
\bal
\n_-\big((\sigma_{AB}^{\ox N})^\G\big)&=\frac{(a_2+c_2)^N-(a_2-c_2)^N}{2},\\
\n_+\big((\sigma_{AB}^{\ox N})^\G\big)&=\frac{(a_2+c_2)^N+(a_2-c_2)^N}{2}.\\
\eal
\eeq
If $\ine\big((\rho_{AB}^{\ox N})^\G\big)=\ine\big((\sigma_{AB}^{\ox N})^\G\big)$, then 
\beq
\label{eq:ncopyines-3}
\bal
\n_-\big((\rho_{AB}^{\ox N})^\G\big)&=\n_-\big((\sigma_{AB}^{\ox N})^\G\big),\\
\n_+\big((\rho_{AB}^{\ox N})^\G\big)&=\n_+\big((\sigma_{AB}^{\ox N})^\G\big).
\eal
\eeq
From \eqref{eq:numcN2n-3} we have $a_1<c_1$ and $a_2<c_2$. Hence, Eq. \eqref{eq:ncopyines-3} is equivalent to $a_1=a_2$ and $c_1=c_2$. It implies $\ine(\rho_{AB}^\G)=\ine(\sigma_{AB}^\G)$. We obtain a contradiction. Therefore, assertion (ii) holds.

This completes the proof.
\end{proof}

Second we present the proof of Theorem \ref{le:xstate} as follows.

\textbf{Proof of Theorem \ref{le:xstate}.}
Denote by $\rho_X$ the $2\times n$ X-state. The density matrix of an arbitrary $2\times n$ X-state can be parametrized as
\beq
\label{eq:2xNX}
\rho_X=
\bma
\begin{array}{c|c}
M_{11} & M_{12} \\ \hline
M_{12}^\dg & M_{22}
\end{array}
\ema,
\eeq
where 
\beq
\label{eq:2xNX-1}
\bal
M_{11}&=\diag(a_1,a_2,\cdots,a_n), \\
M_{22}&=\diag(b_n,b_{n-1},\cdots,b_1),\\
M_{12}&=
\bma
 0 & \cdots & 0 & r_1 e^{i\t_1} \\
 0 & \cdots & r_2 e^{i\t_2} & 0 \\
 \vdots & \iddots & \vdots  & \vdots  \\
 r_n e^{i\t_n} & \cdots & 0 & 0
\ema,
\eal
\eeq
and for all $j$, $a_j,b_j,r_j$ are non-negative real numbers.
With a proper permutation matrix $P$, we have $P\rho_X P^\dg=\op_{j=1}^n B_j$, where $B_j=
\bma
a_j & r_j e^{i\t_j} \\
r_j e^{-i\t_j} & b_j
\ema$.
Thus the eigenvalues of $\rho_X$ can be formulated as
\beq
\label{eq:2xNX-2}
\bal
\l_j^+ &= \frac{a_j+b_j}{2}+\sqrt{r_j^2+d_j^2},\\
\l_j^- &= \frac{a_j+b_j}{2}-\sqrt{r_j^2+d_j^2},\\
\eal
\eeq
where $d_j=\frac{a_j-b_j}{2}$ for any $j$. Since $\rho$ is positive semidefinite, it follows that $\forall j$, $\l_j^+$ and $\l_j^-$ are non-negative. This is equivalent to 
\begin{eqnarray}
\label{eq:rj<=sqrtajbj}
r_j\leq \sqrt{a_jb_j}, ~\forall j.	
\end{eqnarray}
Since $\rho_X^\Gamma$ is still an X-type matrix, one can similarly formulate the eigenvalues of $\rho_X^\Gamma$ as
\beq
\label{eq:2xNX-3}
\bal
\m_j^+ &= \frac{a_j+b_j}{2}+\sqrt{r_{n+1-j}^2+d_j^2},\\
\m_j^- &= \frac{a_j+b_j}{2}-\sqrt{r_{n+1-j}^2+d_j^2}.\\
\eal
\eeq
It follows from Eq. \eqref{eq:2xNX-3} that $\m_j^+\geq 0,\forall j$, and $\m_j^-$ is negative if $r_{n+1-j}>\sqrt{a_j b_j}$. Using this inequality and \eqref{eq:rj<=sqrtajbj}, the number of negative eigenvalues of $\rho_X^\Gamma$ is that of $r_{n+1-j}$ satisfying 
\begin{eqnarray}
\label{eq:an+1-j}	
\sqrt{a_{n+1-j}b_{n+1-j}}\geq r_{n+1-j}>\sqrt{a_j b_j}.
\end{eqnarray}
To satisfy the inequality, we can exclude the case that $n$ is odd and $j=\lc \frac{n}{2}\rc$. Next, we obtain two inequalities by setting $j=k$ and $j=n+1-k$ for every $k\le\lf \frac{n}{2} \rf$ in \eqref{eq:an+1-j}. One can verify that at most one of the two inequalities holds. So the number of negative eigenvalues of $\r^\G_X$ is at most $\lf \frac{n}{2} \rf$. For a fixed $k\leq \lf \frac{n}{2} \rf$, by choosing proper parameters we can make \eqref{eq:an+1-j} hold if and only if $1\leq j\leq k$. Thus the corresponding X-state given by Eq. \eqref{eq:2xNX} is one whose partial transpose has $k$ negative eigenvalues.
This completes the proof.
\qed



\bibliography{witness}

\end{document}